\documentclass[11pt,letterpaper]{article}

\usepackage{style}

\usepackage{bbm}
\usepackage{algpseudocode}
\usepackage{algorithm}
\usepackage{dsfont}
\usepackage{bbold}
\usepackage{textcase}

\newcommand{\FF}{\mathcal{F}}
\newcommand{\EE}{\mathcal{E}}
\newcommand{\MM}{\mathcal{M}}

\let\Var\undefined\DeclareMathOperator{\Var}{Var}
\newcommand{\FVC}{\mathcal{FC}}

\newcommand{\congest}{\textsc{Congest}\xspace}

\newcommand{\alg}{\textsf{ALG}\xspace}
\newcommand{\supported}{\textsc{Supported}\xspace}

\newcommand{\maxExpDeg}{\ensuremath{\bar\Delta}}
\newcommand{\Ugly}{C}

\crefname{observation}{Observation}{Observations}

\AddCommenter{g}{yellow}
\AddCommenter{d}{white!50!green}
\AddCommenter{k}{white!50!purple}
\AddCommenter{a}{white!90!blue}
\AddCommenter{t}{white!75!orange}

\title{Distributed Stochastic Graph Algorithms}
\date{}
\author{Keren Censor-Hillel\\
	\small Technion, Israel \\
	\small ckeren@cs.technion.ac.il \\
	\and
    Aditi Dudeja\\
    \small The Chinese University of Hong Kong, Shenzhen\\
	\small aditidudeja@cuhk.edu.cn
    \and 
    George Giakkoupis\\
    \small  Inria, France\\
	\small george.giakkoupis@inria.fr
}

\begin{document}
\maketitle
\begin{abstract}
We study stochastic graph optimization problems in a novel distributed setting. As in the standard centralized setting, a random subgraph $G^\ast$ of a known base graph $G$ is realized by including each edge $e$ independently with a known probability $p_e$, and we must solve an optimization problem on $G^\ast$ despite uncertainty about its edges. In the standard setting, to cope with this uncertainty, the algorithm can query any edge of $G$ to learn if the edge exists in $G^\ast$, and its complexity is the number of queried edges. The distributed setting incorporates uncertainty in a natural manner, by having each vertex know only about its own edges in $G^\ast$ (and only communicate over them), and the complexity is measured by the number of synchronous communication rounds.

We establish that distributed stochastic algorithms can be drastically faster than their non-stochastic counterparts and overcome known lower bounds, by showing fast distributed approximation algorithms for maximum  matching, minimum vertex cover, and minimum dominating set. Specifically:

\begin{itemize}
    \item Within two rounds of single-bit messages we obtain a constant approximation for maximum  matching, and a single round suffices for an $O(\log \bar\Delta)$-approximation algorithm for minimum dominating set, where $\bar\Delta$ is the maximum expected degree of $G^\ast$. Our most novel technical contribution is a constant approximation algorithm for minimum vertex cover which uses no communication at all.

    \item Allowing more, but still a constant number of rounds enables better approximations:
    For maximum matching we provide a $(0.68-\epsilon)$-approximation and a $(0.73-\epsilon)$-approximation algorithms for general and bipartite graphs,  respectively, running in $\poly(1/\epsilon)$ rounds; and for minimum vertex cover we give a $(2+\epsilon)$-approximation algorithm running in $\poly(1/\epsilon)$ rounds.
\end{itemize}
Some of these results rely on previous works on the standard stochastic setting, namely, [Behnez\-had et al., SODA 2022] and [Derakhshan and Saneian, ICALP 2025], while others require new algorithmic ideas.

\end{abstract}

\section{Introduction}

Stochastic algorithms address a setting in which a problem needs to be solved over some random sample of a given data set. In particular, there has been extensive research on stochastic graph algorithms \cite{BlumDHPSS15,AssadiKL17,BehnezhadR18,AssadiB19,AssadiKL19,BehnezhadD0HR19,BehnezhadFHR19,BehnezhadD20,BehnezhadDH20,BehnezhadBD22,DerakhshanF23,DughmiKP23,AzarmehrBGR25,DerakhshanS25,BehnezhadBD22,DerakhshanDH23,DerakhshanSX25}, where each edge $e$ of a known base graph $G$ has some known probability $p_e$ of being present in a sampled \emph{realization} $G^\ast$ of $G$, and the graph problem needs to be solved over $G^\ast$. The challenge for a stochastic algorithm is the uncertainty about which edges of $G$ actually get realized in $G^\ast$. To this end, the algorithm can query (non-adaptively) 
any set of edges to learn this information, and the complexity of a stochastic graph algorithm is typically measured by the number of edges it queries.

In this paper, we introduce and study a distributed variant of the centralized stochastic setting: 
Nodes know the base graph $G$ and the edge realization probabilities $p_e$, as in the centralized setting.
In addition, each node knows its realized edges and can exchange messages (only) through these edges,
and the complexity is measured by the number of communication rounds.

This model is motivated by the observation that the traditional assumption in distributed graph algorithms that nodes do not have any information about the graph (except possibly for their set of direct neighbors), does not reflect modern networked systems, such as large-scale data centers or content delivery networks, where the underlying physical infrastructure -- the base graph $G$ -- is often static and known to all participants, and the  challenge is that link availability and node congestion fluctuate according to unpredictable, stochastic processes. 
The model we propose is perhaps one of the simplest and most natural abstractions for leveraging the history of a network with the goal of improving algorithmic complexity. 
In practice, realization probabilities $p_e$ can be viewed as predictions generated by machine learning models trained on historical traffic patterns.
Next we describe the model in detail.

\paragraph{Distributed Stochastic Model.}
In the classic distributed models, $n$ vertices in a network communicate in synchronous rounds, in each of which every vertex can send a message to each of its neighbors. When considering distributed algorithms for a stochastic setting, we follow the paradigm of the centralized stochastic setting, which allows the algorithm to preprocess the base graph.
Thus, in the distributed stochastic model, the base graph $G=(V,E)$ (i.e., the network) is known to all vertices. For a graph problem $P$, its stochastic version is the problem $P$ over a random realization $G^\ast$, which is a subgraph of $G$ obtained by including every edge $e\in E(G)$ in it independently with probability $p_e$, where the values $p_e$ are known to all vertices. The uncertainty about which edges actually appear in $G^\ast$ is incorporated into the distributed model in a natural manner: each vertex $v$ knows only the set $N_{G^\ast}(v)$ of its neighbors in $G^\ast$, and can communicate only over these edges. A distributed stochastic algorithm may preprocess the base network graph $G$, and its complexity is its number of communication rounds after $G^*$ is realized.\footnote{More generally, instead of knowing the complete base graph $G$, one may allow each vertex $v$ to know (and preprocess) only the subgraph of $G$ induced by $v$'s $h$-hop neighborhood in $G$. We believe that some of our algorithms can be adapted to this more refined model, for $h = \poly \log(n)$.}

We demonstrate that the distributed stochastic model is remarkably powerful, allowing us to bypass lower bounds that apply to classical distributed algorithms. 
Namely, we provide distributed approximation algorithms for maximum matching, minimum vertex cover, and minimum dominating set, that require a constant number of  communication rounds (and sometimes no communication at all). 
Achieving comparable approximations by classical distributed algorithms requires a super-constant number of communication rounds.

\subsection{Our Contribution}

For stochastic algorithms, an $\alpha$-approximation refers to the ratio between the expected size of the solution and the expected size of an optimal solution (both for $G^\ast$).

The round complexity of our algorithms is either a small constant, or a constant that depends on an approximation constant $\epsilon$.
Notably, the round complexity does not depend on $n$ nor on the realization probabilities $p_e$.
In contrast, in the centralized setting it is known that that the total number of queries must depend on both (as well as on $\epsilon$) \cite{AssadiKL19}.
 All of our algorithms use messages of size at most $O(\log{n})$ bits as in the distributed \congest model, and many of them only need single-bit messages.
 
For all three optimization problems, our results can be contrasted with the non-stochastic distributed setting, in which any polylogarithmic approximation needs a number of rounds that is $\Omega(\log\Delta/\log\log\Delta)$ or $\Omega(\sqrt{\log{n}/\log\log{n}})$ \cite{KuhnMW16}.

\paragraph{Minimum Vertex Cover.} Our most novel technical contribution is in a constant approximation for minimum vertex cover, which can be obtained without communication.

\begin{restatable}{theorem}{VCnoRounds}
\label{thm:VCnoRounds}
    There is a deterministic algorithm for the distributed stochastic minimum vertex cover problem that outputs a $3.44$-approximation and requires no communication between vertices.
    
\end{restatable}

Our algorithm, given in \Cref{sec:MVCnoCommunication}, \emph{associates} each edge $uv$ of $G$ with its endpoint, $u$ or $v$, that has the larger conditional probability of being in the optimal vertex cover of $G^\ast$, given that $uv$ is realized. (These probabilities can be computed from $G$ and the edge realization probabilities in a preprocessing phase.)
Note that the larger of these two probabilities must be at least $1/2$.
A vertex is then included in the vertex cover if and only if at least one of its associated edges is realized.
For intuition, suppose that all edges have the same realization probability $p$, and let $E_v$ denote the set of edges associated with vertex $v$.
The probability that $v$ is in the optimal vertex cover is bounded from below by the minimum probability with which $v$ must be in a vertex cover to ensure that 
for each of the (independently realized) edges $e\in E_v$, $v$ covers $e$ with probability at least $1/2$.

This minimum probability can be computed using properties of the binomial distribution.
It is at most a constant factor smaller than the probability $1-(1-p)^{|E_u|}$ that $v$ covers every realized edge $e\in E_v$ with probability $1$, which is the probability that $v$ is in the vertex cover output by our algorithm.
To handle the case of different edge realization probabilities, we use a refinement of this argument, which involves associating with each edge an independent Poisson random variable that is non-zero if and only if the edge is realized, and is used instead of a binary random variable.

We also show (in \Cref{sec:BetterApxMVC}) that we can obtain a $(2+\epsilon)$-approxi\-mation within $\poly(1/\epsilon)$ rounds.

\begin{restatable}{theorem}{VCInPolyRrounds}
\label{thm:VCInPolyRrounds}
There is a deterministic algorithm for the distributed stochastic minimum vertex cover problem that outputs a $(2+\epsilon)$-approximation in $\poly(1/\epsilon)$ rounds of single-bit messages.

\end{restatable}

To obtain this result, we follow an approach from \cite{BehnezhadBD22}.
At the core of this approach is a standard algorithm that computes a fractional matching by uniformly increasing the fractional weights of all edges;  
an edge becomes inactive once the total weight of the edges at one of its endpoints reaches 1.
This algorithm starts on the base graph $G$, and switches to the realized graph $G^\ast$ as soon as the maximum edge weight reaches a threshold.
Adapting the approach to the distributed setting presents some technical challenges.
In particular, uniformly increasing the edge weights in the first stage of the algorithm results in prohibitively large remaining vertex degrees in the second stage. 
We solve this problem by increasing the edge weights proportionally to the realization probabilities, rather than uniformly.
This ensures that vertices have constant expected remaining degrees in the realized graph, but we still need to handle separately vertices whose actual degree is super-constant.
Moreover, combining the two parts of the algorithm requires more care over the original approach in \cite{BehnezhadBD22}. 

We mention that paper \cite{DerakhshanDH23} gives an algorithm for centralized stochastic minimum vertex cover that obtains a $(3/2+\epsilon)$ approximation, which is better than ours. 
However, we are not aware of a way to adapt the approach given there in our setting.
The main obstacle is that the approach needs to compute a near-optimal vertex cover in a sparse subgraph of the base graph.
In our setting we can obtain the sparse subgraph, but do not know how to compute the near-optimal vertex cover efficiently.

\paragraph{Maximum Matching.}
For the maximum matching problem, we can get a constant approximation in only two rounds (see \Cref{sec:mmTwoRounds}).

\begin{restatable}{theorem}{constMatchingTwoRounds}
    \label{thm:constMatchingTwoRounds}
    There is a randomized algorithm for the distributed stochastic maximum matching problem that outputs a $0.39$-approximation in two rounds of single-bit messages.
    If the base graph is bipartite, an improved approximation ratio of $0.63$ is achieved. 
    
\end{restatable}

The algorithm for bipartite graphs is identical to an algorithm given in \cite{BehnezhadBD22} for the centralized stochastic setting: 
Each vertex from the one side of the bipartite graph proposes to a neighbor on the other side, chosen independently at randomly according to the matching probabilities from the optimal solution; and each vertex from the receiving side that gets at least one proposal accepts one of them.
For non-bipartite graphs we essentially apply the same algorithm to a random bipartite subgraph, obtained by choosing the side of each vertex independently at random (with near uniform probability). 
Making sure that the proposal probabilities are independent, as in the bipartite case, requires some care in handling non-bipartite edges.

We show (in \Cref{sec:matching-const-approx-polyeps}) that we can achieve better constant approximations using $\poly(1/\epsilon)$ rounds.

\begin{restatable}{theorem}{MatchingInPolyRrounds}
\label{thm:MatchingInPolyRrounds}

    There is a randomized algorithm for the distributed stochastic maximum matching problem that outputs a $(0.68-\epsilon)$-approximation in $\poly(1/\epsilon)$ rounds of $O(\log n)$-bit messages.
    If the base graph is bipartite, an improved approximation ratio of $(0.73-\epsilon)$ is achieved. 
    Both results assume that all edges $e$ have the same realization probability $p_e=p$.\footnote{This is the only result in this paper in which we make this assumption.} 
\end{restatable}

To prove the above, we rely on a result from \cite{DerakhshanS25} which states that there exists a subgraph $Q$ of $G$ of maximum degree $p^{-1}\cdot\poly(1/\epsilon)$, whose realization $Q^\ast$ has expected matching size close to the matching size of $G^\ast$.\footnote{In \cite{DerakhshanS25}, $p=\min_{e\in E} p_e$ is the minimum realization probability of any edge.} 
In our distributed setting, this subgraph $Q$ can be computed in the preprocessing stage. For our algorithm, we assume that all edges $e$ have the same realization probability $p_e=p$,
thus the expected degree of any vertex in $Q^\ast$ is $\poly(1/\epsilon)$.
If this were the case for the \emph{maximum} degree of $Q^\ast$, we could apply known distributed algorithms to obtain near optimal matching in $\poly(1/\epsilon)$ rounds. However, the maximum degree in $Q^\ast$ may be higher. We overcome this issue by proving that removing vertices of higher degree decreases the matching size only marginally, and we then apply a known distributed matching algorithm on the remaining graph.

\paragraph{Minimum Dominating Set.} 
We show (in \Cref{sec:mds}) that one round of communication allows an approximation that is logarithmic in the maximum expected degree of $G^\ast$, 
\[
    \maxExpDeg= \max_{v\in V}  \Exp[\deg_{G^\ast}(v)].
\]

\begin{restatable}{theorem}{MDSlogdelta}
\label{thm:MDSlogdelta}
    There is a deterministic algorithm for the distributed stochastic minimum dominating set problem that outputs a $O(\log{\maxExpDeg})$-approximation in a single round of single-bit messages.

\end{restatable}

Our approach is to mimic the greedy algorithm using the expected degrees in $G^\ast$ rather than the actual degrees. However, a technical challenge arises since actual degrees may differ from their expectation. It is relatively straightforward to show a bound of $O(\log n)$ on the approximation ratio.
To prove the logarithmic bound on $\bar \Delta$ we use a careful argument that accounts for the contribution of vertices whose actual degree exceeds their expectation or falls short of the expectation. 

\subsubsection{An Improved 3-Approximation for Vertex Cover without Communication}

A result from recent independent work \cite{Brand2026} immediately implies a $2c$-approximate distributed stochastic vertex cover algorithm without communication, where $c<8$.
Their result assumes $p_e=p$ for all edges $e$.
Precisely, \cite[Lemma 4.4]{Brand2026} shows that there is a set of vertices $S\subseteq V$ such that $|S| \leq c\cdot \Exp[|\mathrm{MVC}(G^\ast)|]$ and the number of edges in the remaining induced subgraph $G[V\setminus S]$ is at most $c\cdot \Exp[|\mathrm{MVC}(G^\ast)|]/p$, where $\mathrm{MVC}(H)$ denotes a minimum vertex cover of any graph $H$. 
Therefore, a set containing all vertices from $S$ plus one endpoint of each edge in $G^\ast[V\setminus S]$ forms a vertex cover of $G^\ast$ with an expected size of at most $2c \cdot \Exp[|\mathrm{MVC}(G^\ast)|]$.

Using the construction from \cite[Lemma 4.4]{Brand2026}, we can obtain an improved $3$-approximate distributed stochastic algorithm without communication.
By a greedy process, the lemma constructs an ordering $\pi = (v_1,\ldots,v_n)$ of the vertices that satisfies the following property:
Let $M$ be a random matching of $G^\ast$ obtained by a sequential process considering the vertices in the order of $\pi$. 
When vertex $v_i$ is considered, if it is not already matched (to a vertex in $\{v_{1},\ldots,v_{i-1}\}$) and has at least one unmatched neighbor in $\{v_{i+1},\ldots,v_n\}$, then $v_i$ is matched to a neighbor chosen uniformly at random from those options.
Let $q_i$ be the probability that $v_i$ is matched to a vertex in $\{v_{1},\ldots,v_{i-1}\}$.
The ordering $\pi$ ensures that the probability of $v_i$ being matched to a vertex in $\{v_{i+1},\ldots,v_{n}\}$ is at least $R_i (1-2q_i)$, where $R_i$ is the probability that $v_i$ has at least one neighbor in $\{v_{i+1},\ldots,v_{n}\}$ in $G^\ast$.

Consider now the (random) vertex set $C$ consisting of all vertices $v_i$ that have at least one neighbor in $\{v_{i+1},\ldots,v_{n}\}$ in $G^\ast$. 
Note that $C$ is a valid vertex cover of $G^\ast$ because for any realized edge $(v_i, v_j)$ with $i < j$, its endpoint $v_i$ is included in $C$. 
Moreover, $C$ can be constructed without communication given $\pi$.
Its expected size is $\Exp[|C|] = \sum_{i=1}^n R_i$.
We can bound the expected size of matching $M$ by counting its edges via their endpoints. Counting via the higher-indexed (right) endpoints yields $\Exp[|M|] = \sum_{i=1}^n q_i$, while counting via the lower-indexed (left) endpoints yields:
\[
    \Exp[|M|] 
    \geq 
    \sum_{i=1}^n R_i(1-2q_i)
    \geq
    \sum_{i=1}^n R_i 
    -
    \sum_{i=1}^n 2q_i
    =
    \Exp[|C|] - 2\Exp[|M|]
    .
\]
Rearranging yields $\Exp[|C|] \leq 3\Exp[|M|]$. Since the size of any matching is a lower bound for the size of a minimum vertex cover, it follows that $\Exp[|C|] \leq 3\Exp[|\mathrm{MVC}(G^\ast)|]$.
This result does not require the assumption that all edges have the same realization probability $p$.

\subsection{Related Work}
\label{sec:related}

We restrict our attention to three areas of related work: algorithms for centralized stochastic graph problems, distributed graph algorithms for the standard \congest and related models, and distributed graph algorithms with predictions and advice.   

\paragraph{Centralized Stochastic Graph Algorithms.}

As mentioned earlier, the limiting resource in these works is the number edges that the algorithm needs to query (non-adaptively), in order to solve an approximate optimization problem on the realized graph.   
Motivated by maximizing the expected number of successful transplants in a kidney exchange problem, the pioneering work of \cite{BlumDHPSS15} proposed a constant-approximate stochastic matching algorithm using a number of queries per vertex that depends only on the (minimum) edge  realization probability.
This has led to a major line of research about improving the approximation ratio and reducing the number of queries \cite{AssadiKL17,BehnezhadR18,AssadiB19,AssadiKL19,BehnezhadD0HR19,BehnezhadFHR19,BehnezhadD20,BehnezhadDH20,BehnezhadBD22,DerakhshanF23,DughmiKP23,AzarmehrBGR25,DerakhshanS25}.
More recently, stochastic minimum vertex cover has been studied \cite{BehnezhadBD22,DerakhshanDH23,DerakhshanSX25,Brand2026}.
Other problems that have been considered include 
stochastic minimum spanning trees \cite{GoemansV06}, stochastic packing problems \cite{MaeharaY20}, and stochastic shortest path problems \cite{vondrak2007}.

\paragraph{Related work in \textnormal{\congest} and \textnormal{\supported} Models.}


Exact computation of the minimum vertex cover or dominating set in the \congest model requires $\tilde{\Omega}(n^2)$ rounds in general graphs~\cite{AbboudCKP21,BachrachCDELP19}. Approximation algorithms have been extensively studied for 
maximum matching \cite{Luby85, IsraelI86, AlonBI86, WattenhoferW04,LotkerPR09, LotkerPP15, Bar-YehudaCGS17, AhmadiKO18, Ben-BasatKS19, Fischer20, FaourFK21, FischerMU22,ChangS22, KitamuraI22, HuangS23, IzumiKY24, MitrovicS25},
minimum vertex cover ~\cite{KhullerVY94,HanckowiakKP01,PanconesiR01,KuhnMW06,GrandoniKP08,KoufogiannakisY09,AstrandFPRSU09,PolishchukS09,AstrandS10,BarenboimEPS12,Fischer17,Bar-YehudaCS17,Ben-BasatEKS18,FaourK20,FaourFK21},  and minimum dominating set~\cite{JiaRS02, KuhnW05, KuhnMW06, KuhnMW16, GhaffariK18, DeurerKM19, Censor-HillelD21, MorganSW21, Faour0GKR23, DoryGI24}.
However, an $\Omega(\log\Delta/\log\log\Delta)$ or $\Omega(\sqrt{\log{n}/\log\log{n}})$ lower bound is known for polylogarithmic approximations~\cite{KuhnMW04}. For vertex cover, higher logarithmic lower bounds are known for $(1+\epsilon)$-approximations~\cite{GoosS14,FaourFK21}.


In the \supported model \cite{SchmidS13}, the topology of the communication graph is known in advance and a graph problem needs to be solved for an adversarially chosen subgraph \cite{KorhonenR17, FoersterKR019, FoersterH0S19, FaourFK21, FoersterKPR021, HaeuplerWZ21, KorhonenPRSS21, BalliuKKLOPPR0S23, AnagnostidesLHZG23, AgrawalAPR23, BalliuB0O24}. While this model also leverages processing, our setting differs in that (i) the input subgraph is random rather than adversarial, and (ii) the communication takes place only over edges of the subgraph rather than over the entire base graph (this is referred to as a passive \supported model in \cite{FoersterH0S19}).

\paragraph{Distributed Graph Algorithms with Predictions and Advice.}

Our model 
can be viewed as part of a general paradigm where nodes are given extra information about the graph. 
In distributed algorithms with predictions, each node is 
also
given some extra information that may (or may not) be correct, and the goal is for the algorithm to work faster when predictions are good and to not work much worse than without predictions when predictions are bad.
Distributed maximal independent set algorithms with predictions, where each node is given a prediction for its output, were recently studied in \cite{Boyar2025}.
In distributed algorithms with advice, each node is given some additional information, as well, but this information is always correct and is chosen by the algorithm designer. 
The goal is to measure the improvement possible for a given advice size, or compute the minimum advice size to significantly improve performance.
Distributed graph algorithms with advice have been proposed for coloring and minimum spanning trees \cite{FraigniaudGIP09,FraigniaudKL10,Balliu0K0ORS25}.

\section{Constant Approximation for Minimum Vertex Cover without Communication}

\label{sec:MVCnoCommunication}

We present an algorithm that computes a constant approximation of the minimum vertex cover without requiring any communication.

\VCnoRounds*

Let $\FVC$ be any fractional vertex cover algorithm, and let $\FVC(H)$ denote its output for input graph $H$.\footnote{A fractional vertex cover is an assignment of values in $[0,1]$ to the vertices such that for every edge, the sum of values of both endpoints is at least 1. The value of the cover is the sum of all vertex values. We assume a fractional vertex cover algorithm instead of an integral one, because an optimal fractional vertex cover is computable in polynomial time.}
We denote by $F_{v}$ the fractional weight assigned to vertex $v$ in fractional cover $\FVC(G^\ast)$ of the realized graph $G^\ast=(V,E^\ast)$.
For every edge $vu$ of $G$ we denote by $f_{vu}$ the conditional expectation of $F_v$ given that $vu$ is realized,
\[
    f_{vu} = \Exp[F_v \mid vu \in E^\ast]
    .
\]
Note that $f_{vu}+f_{uv} \geq 1$ and thus $\max\{f_{vu},f_{uv}\}\geq 1/2$.

\paragraph{Distributed Algorithm}

\begin{enumerate}
    \item (\textbf{Preprocessing stage})
    Each vertex $v\in V$ computes a subset $E_v$ of its incident edges in $G$, where
    \[
        E_v = \{vu \in E(G) \colon f_{vu} \geq f_{uv}\}
        .
    \]
    This is the set of edges that $v$ is `responsible' for covering if realized.
    
    \item (\textbf{Distributed stage}) At this stage, each $v\in V$ is made aware of the set of its realized incident edges in $G^\ast$.
    If 
    at least one edge $e\in E_v$ is realized, then $v$ is added to the output vertex cover set, denoted $C$.
\end{enumerate}

Correctness follows from the fact that for each edge $uv$,  $uv\in E_u$ or $uv\in E_v$, and if $uv$ is realized then the responsible vertex is included in the vertex cover.  
The approximation ratio of $3.44$ follows from the next lemma.

\begin{lemma}
    \label{lem:vc-const-approx}
    For every vertex $v\in V$,
    the probability that $v$ is contained in the vertex cover $C$ computed by the distributed algorithm is
    \[
        \Pr(v \in C) \leq 3.44 \cdot\Exp[F_v]
        .
    \]
\end{lemma}

From \cref{lem:vc-const-approx} we get that the expected size of the vertex cover $C$ is at most $3.44$ times larger than the expected size of fractional cover $\FVC(G^\ast)$, i.e.,
\[
    \Exp[|C|] \leq 3.44\cdot \sum_{v\in V} \Exp[F_v]
    .
\]

\begin{proof}[Proof of \cref{lem:vc-const-approx}.]
Let $e=uv$. For every edge $e\in E_v$, let $X_e$ be the indicator random variable of the event $vu\in E^\ast$, that edge $vu$ is realized in $G^\ast$.
   
    We associate with each edge $e\in E_v$ a Poisson random variable $Y_e$ with mean 
    \[
        \lambda_e = \ln \frac1{1-p_e}
        ,
    \]
    such that:
    \begin{enumerate}
        \item $Y_e = 0$ if and only if $X_e =  0$,
        \item The random variables $Y_e$, for $e\in E_v$, are mutually independent, and
        \item $Y_e$ is conditionally independent of $F_v$ given $X_e$.
    \end{enumerate}
    The first point is feasible because
\begin{equation*}
        \Pr(Y_e = 0) =  e^{-\lambda_e} 
        = e^{-\ln \frac1{1-p_e}} 
        = 1-p_e = \Pr(X_e=0).
 \end{equation*}       
    All three points can be satisfied simply by drawing the values of the $Y_e$ \emph{after} the realized graph $G^\ast$ and the cover $\FVC(G^\ast)$ are determined, such that $Y_e=0$ if $X_e=0$, while if $X_e=1$, $Y_e$'s value is drawn independently from the conditional Poisson distribution of mean $\lambda_e$ given $Y_e\neq0$.
    Let
    \[
        Y_v = \sum_{e\in E_v} Y_e.
    \]
    Being the sum of independent Poisson random variables, $Y_v$ is a Poisson random variable as well, with mean
    \[
        \lambda_v 
        = 
        \sum_{e\in E_v} \lambda_e.
    \]
    We can express $\Pr(v\in C)$ in terms of $Y_v$, as
    \begin{equation}
        \label{eq:prvC}
        \Pr(v\in C)
        =
        \Pr\left(\sum_{e\in E_v} X_e \neq 0\right)
        =        
        \Pr\left(\sum_{e\in E_v} Y_e \neq 0\right)
        =
        \Pr(Y_v\neq 0).
    \end{equation}
where the second equation holds because $Y_e = 0$ if and only if $X_e =  0$. 
Next we show that $\Exp[F_v\cdot Y_v] \geq \frac12\cdot \Exp[Y_v]$. 
   
    For any edge $e = vu \in E_v$, 
    \begin{align*}
        \Exp[F_v \cdot Y_e]
        &=
        \Exp[F_v \cdot Y_e \mid Y_e \neq 0]\cdot \Pr(Y_e \neq 0)
        \\&
        =
        \Exp[F_v \cdot Y_e \mid X_e = 1]\cdot \Pr(Y_e \neq 0)
        \\&
        =
        \Exp[F_v \mid X_e = 1]\cdot \Exp[Y_e \mid X_e = 1]\cdot \Pr(Y_e \neq 0)
        \\&
        =
        f_{vu}\cdot \Exp[Y_e \mid Y_e \neq 0]\cdot \Pr(Y_e \neq 0),    
        \\&
        =
        f_{vu}\cdot \Exp[Y_e]    
        \\&
        \geq
        \frac12\cdot \Exp[Y_e],    
    \end{align*}
    where 
    in the third line we used the conditional independence of $Y_e$ and $F_v$ given $X_e$,
    and in the last line we used that $f_{vu}\geq 1/2$ because $f_{vu} \geq f_{uv}$ and $f_{vu} + f_{uv}\geq 1$.
    It follows
    \begin{equation}
        \label{eq:expFvYv}
        \Exp[F_v \cdot Y_v]
        =
        \sum_{e\in E_v}\Exp[F_v \cdot Y_e]
        \geq
        \frac12\cdot \sum_{e\in E_v}\Exp[Y_e]
        =
        \frac12\cdot \Exp[Y_v].
    \end{equation}
    Inequality \cref{eq:expFvYv} allows us to apply a  property of the Poisson distribution that we show in technical \cref{lem:poisson-lemma}, to obtain
    \[
        \Exp[F_v] \geq \frac{\Pr(Y_v \neq 0)}{3.44}
    \]
    Combining this and \cref{eq:prvC} completes the proof of the lemma.
\end{proof}

\begin{lemma}[A Property of Poisson Distribution]
    \label{lem:poisson-lemma}
    Let $Y$ be a Poisson random variable, and $F$ be a real random variable that takes values in the interval $[0,1]$.
    The two random variables are not necessarily independent.
    If $\,\Exp[F\cdot Y] \geq \frac12\cdot \Exp[Y]$ then 
    \[
        \Exp[F] \geq \frac{\Pr(Y \neq 0)}{3.44}
        .
    \]
\end{lemma}
\begin{proof}
    Let $\lambda = \Exp[Y]$ be the mean of Poisson random variable $Y$, and for $k\geq 0$, let $q_k$ denote the probability that $Y=k$,
    \[
        q_k = \Pr(Y=k) = \frac{\lambda^ke^{-\lambda}}{k!}
        .
    \]
    Let $m \geq 1$ be an integer and $\beta\in (0,1]$ a real such that 
    \begin{equation}   
        \label{eq:zeta-beta}
        \sum_{k\leq m} k q_k
        + \beta (m+1) q_{m+1} = \frac\lambda 2
        .        
    \end{equation}
    (We will see that  $m$ is the median of $Y$.)    
    Let $F^\ast$ be a random variable such that
    \[
        F^\ast = 
        \begin{cases}
            0 & \text{if $0\leq Y\leq m$}\\ 
            1-\beta & \text{if $Y = m+1$}\\ 
            1 & \text{if $Y \geq m+2$}\\ 
        \end{cases}
    \]
    We have that $F^\ast$ is in $[0,1]$, and we will show that $\Exp[F^\ast\cdot Y] = \frac\lambda 2$. We then argue that among all random variables $F$ that take values on $[0,1]$ and satisfy $\Exp[F\cdot Y] \geq \frac \lambda 2$,
    random variable $F^\ast$ has the smallest mean. 
    I.e., $\Exp[F^\ast] \leq \Exp[F]$, for any random variable $F$ that takes values on $[0,1]$ and $\Exp[F\cdot Y] \geq \frac \lambda 2$. Thus, to prove the lemma it suffices to show 
    \[
        \Exp[F^\ast] 
        \geq 
        \frac{\Pr(Y\neq 0)}{3.44}
        .
    \]
    We begin by showing that $\Exp[F^\ast\cdot Y] = \frac\lambda 2$:
    \[
      \begin{aligned}
        \Exp[F^\ast\cdot Y] 
        &=
        (1-\beta) (m+1) q_{m+1}
        +
        \sum_{k\geq m+2} k q_k
        \\&
        =
        (1-\beta) (m+1) q_{m+1}
        +
        \Exp[Y] -
        \sum_{k\leq m+1} k\cdot q_k
        \\&
        =
        \Exp[Y] - \frac\lambda 2
        \\&
        =
        \frac\lambda 2
        ,
      \end{aligned}
    \]
    where the second-last line was obtained using \cref{eq:zeta-beta}.

    Next, we explain why $\Exp[F^\ast] \leq \Exp[F]$, for any $F$ that takes values on $[0,1]$ and satisfies that $\Exp[F\cdot Y] \geq \frac \lambda 2$:
    For $y\geq 0$, let $z_y = \Exp[F\cdot \mathds{1}_{Y=y}] \leq q_y$.
    The constraint gives us that $\sum_y y\cdot z_y \geq \lambda/2$, and the goal is to minimize $\sum_y z_y$.
    Since $y$ is monotonically increasing, by the rearrangement inequality, the minimum for $\sum_y z_y$ is attained by maximizing values of $z_y$ for larger $y$, namely, setting  $z_y = \Exp[F\cdot Y] = q_y$, or equivalently, setting $F=1$ if $y>y_0$, and setting $F=0$ if $y<y_0$, for some $y_0$.

   It remains to show
    \[
        \Exp[F^\ast]
        \geq 
        \frac{\Pr(Y\neq 0)} {3.44}
        .
    \]
    First, notice that
    \begin{equation}
        \label{eq:expfast}
        \Exp[F^\ast]
        =
        (1-\beta) q_{m+1} + \sum_{k\geq m+2} q_k
        =
        1- \sum_{k\leq m}q_k - \beta q_{m+1}
        .
    \end{equation}
    Substituting the definition of $q_k$ to \cref{eq:zeta-beta} and simplifying gives
    \[
        \sum_{1\leq k\leq m} \frac{\lambda^ke^{-\lambda}}{(k-1)!}
        + \beta \frac{\lambda^{m+1}e^{-\lambda}}{m!} = \frac\lambda 2
        ,        
    \]
    thus
    \begin{equation}
        \label{eq:zmedian}        
        \sum_{k\leq m-1} q_k
        +\beta q_{m} = \frac1 2
        .
    \end{equation}
    Since $\beta\in(0,1]$, the above implies that $m$ is the smallest integer such that $\sum_{k\leq m} q_k\geq \frac12$, i.e., $m$
    is the median of $Y$.
    For every $i \geq 0$, let $\lambda_{i+1}$ be the value of $\lambda$ that satisfies
    \[
        \sum_{k\leq i} q_k = \frac1 2;
    \]
    let also $\lambda_{0}=0$.
    Then, for $\lambda_{i} < \lambda \leq \lambda_{i+1}$ the median is $m=i$.

    We first use \cref{eq:zmedian} to compute a loose bound on $\Exp[F^\ast]$, which is however enough to prove the lemma for most values of $\lambda$. 
    Combining \cref{eq:expfast,eq:zmedian}, we obtain
    \[
        \Exp[F^\ast]
        =
        \frac12 - (1-\beta)q_{m} - \beta q_{m+1}
        =
        \frac12 - q_{m} +\beta(q_{m}-q_{m+1})
        \geq 
        \frac12 - q_{m},
    \]
    where the last inequality holds because 
    \[
        q_{m+1} = q_m\cdot\frac{\lambda}{m+1} < q_m,
    \]
    since $m \geq \lambda + \ln(2)$~\cite{Choi1994} and thus $m+1 > \lambda$.
    
    For any $i\geq1$ and $\lambda_{i} < \lambda \leq \lambda_{i+1}$, the above inequality gives
    \[
        \Exp[F^\ast]
        \geq
        \frac12 - q_i
        =
        \frac12 - \frac{\lambda^ie^{-\lambda}}{i!}
        \geq
        \frac12 - \frac{i^ie^{-i}}{i!}
        ,
    \]
    since $\lambda^ie^{-\lambda}$ is minimized for $\lambda=i$. 
    Note also that that term $\frac{i^ie^{-i}}{i!}$ decreases as $i$ increases.
    Then, for $i\geq 4$, 
    \[
        \Exp[F^\ast]
        \geq 
        \frac12 - \frac{4^4e^{-4}}{4!}
        > 0.304 >
        \frac1{3.3}.
    \]
    Therefore the lemma holds for $\lambda > \lambda_4$. 
    
    For the case of $\lambda\leq \lambda_4$, we analyze the precise formula for $\Exp[F^\ast]$.
   
    We substitute in \cref{eq:expfast} the value of $\beta$ obtained by solving \cref{eq:zeta-beta} for $\beta$, to get
    \[
        \Exp[F^\ast]
        =
        1- \sum_{k\leq m}q_k 
        -
        \frac\lambda{m+1}\left(\frac12 - \sum_{k\leq m-1} q_k \right)
        .
    \]
    For $\lambda \leq \lambda_1$ (and thus $m=0$), we have $\Exp[F^\ast]=1-\lambda/2$. This gives
    \[
        \frac{\Exp[F^\ast]}{\Pr(Y\neq 0)}
        =
        \frac{1-e^{-\lambda}-\lambda/2}{1-e^{-\lambda}},
    \]
    which is a decreasing function on $\lambda$, thus the minimum is at the right endpoint $\lambda_1$ of the interval.
    For $\lambda_i < \lambda \leq \lambda_{i+1}$, for each $i\in\{1,2,3\}$, we have verified that the function   
    \[
        \frac{\Exp[F^\ast]}{\Pr(Y\neq 0)}
        =
        \frac{1- \sum_{k\leq i}q_k 
        -
        \frac\lambda{i+1}\left(\frac12 - \sum_{k\leq i-1} q_k \right)}{1-q_0}
    \]
    is concave in each of the three intervals (the second derivative is negative), which implies the minimum lies in one of the points $\lambda_i$, for $i\in\{1,2,3,4\}$. 
    A plot of the function is shown in \cref{fig:plot}. 
    Recall that $\lambda_{i+1}$ is the value of $\lambda$ that satisfies
    $\sum_{k\leq i} q_k = \frac1 2$.
    Solving that and substituting above we obtain that the global minimum is at point 
    $\lambda_2 = 1.678347$ and the minimum value is $\frac1{3.43068}$.
\end{proof}

\begin{figure}[t]
    \centering
    \includegraphics[width=0.5\linewidth]{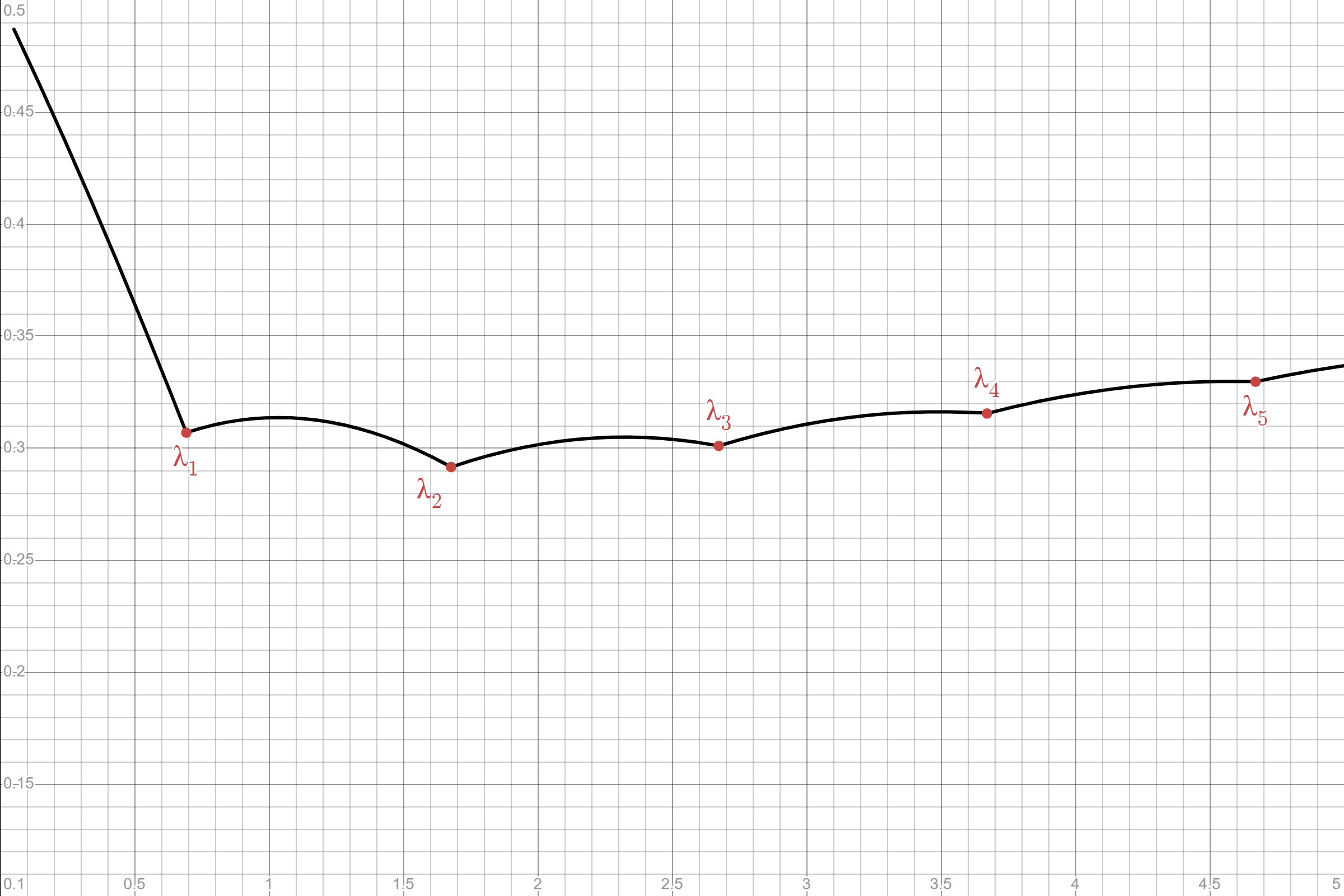}
    \caption{Plot of $\frac{\Exp[F^\ast]}{\Pr(Y\neq 0)}$ as a function of $\lambda$, used in the proof of \cref{lem:poisson-lemma}.}
    \label{fig:plot}
\end{figure}

\section{Improved Approximation for Minimum Vertex Cover in \texorpdfstring{$\poly(1/\eps)$}{poly(eps){-1}} Rounds}
\label{sec:BetterApxMVC}

In this section we show that we can improve the approximation factor to $2+\epsilon$ at the cost of not too many communication rounds, namely, $\poly(1/\eps)$. We give an algorithm that finds a vertex cover of $G^\ast$ with the following properties.

\VCInPolyRrounds*

Our approach is based on an algorithm from \cite
{BehnezhadBD22}, but requires several modifications, as mentioned earlier. 
Throughout the algorithm we will use the following constants which we state here for convenience. 
Let $0<\varepsilon \leq\frac{1}{4}$, and let 
\begin{gather*}
    \varepsilon_1 = \varepsilon^3
    \qquad
    \varepsilon_2 = \varepsilon + \varepsilon^3
    \qquad
    \varepsilon_3 = \varepsilon - \varepsilon^3
    \\
    \xi = \frac{1+\varepsilon_2}{\varepsilon_1} = \frac{1+\varepsilon+\varepsilon^3}{\varepsilon^3}
    \\
    \varepsilon_4 = \varepsilon_2 +\varepsilon_3
    = 2\varepsilon
    \qquad
    \varepsilon_5 = \frac{\varepsilon_1}{(\varepsilon_2 - \varepsilon_1)^2}
    = \varepsilon
    \\
    \epsilon 
    =
    \frac{(2 + \varepsilon_5)(1+\varepsilon_4)}{1 - \varepsilon_5} - 2
    =
    \frac{(2+\varepsilon)(1 + 2\varepsilon)}{1 - \varepsilon} - 2
    \leq 10\varepsilon
    ,
\end{gather*}
where for the last inequality we used that $\varepsilon\leq 1/4$.

\paragraph{Distributed Algorithm}

\begin{enumerate}
    \item {\bf(Preprocessing Stage)} 
    Run \Cref{alg:waterfilling} on $G=(V,E)$ to obtain a fractional matching $\phi=(\phi_e)_{e\in E}$ of $G$. 
    \Cref{alg:waterfilling} is similar to the algorithm in \cite{BehnezhadBD22}, except that the fractional edge weights change proportionally to their realization probability, instead of uniformly.
    For each vertex $v$, let $\phi_v = \sum_{e\ni v}\phi_e$, and let
    \[
        F = \{v\in V \colon \phi_v = 1\}.
    \]

    \item {\bf(Distributed Phase)} Let $G^\ast = (V,E^\ast)$ denote the realized graph. Define $\chi=(\chi_e)_{e\in E}$ as 
    \[
        \chi_e = 
        \begin{cases}
            \frac{\phi_e}{p_e}, & \text{if } e\in E^\ast\\
            0, & \text{if } e\in E\setminus E^\ast.    
        \end{cases}
    \]
    
    Let $\chi_v = \sum_{e\ni v}\chi_e$, for $v\in V$. Let the set of \emph{bad} vertices be 
    \[
        B = \left\{v\in V \colon  \chi_v \geq \phi_v + \varepsilon_2  \right\}.
    \]
    Let $Q$ be the set of all edges of $G$ with no endpoints in $F\cup B$,
    \[
        Q = \{e\in E \colon e \cap (F\cup B) =\emptyset\},
    \]
    and let $Q^\ast = Q\cap E^\ast$. We run \Cref{alg:distributedwaterfilling} in a distributed manner on graph $(V,Q^\ast)$ with $b=(\phi_v)_{v\in V}$ and obtain $\psi=(\psi_e)_{e\in Q^\ast}$. For convenience, we also define $\psi_e = 0$, for all $e\in E\setminus Q^\ast$.
    \Cref{alg:distributedwaterfilling} is a discretized version of the algorithm in \cite{BehnezhadBD22}. 
    The final output of the algorithm is the set 
    \[
        C=F\cup B\cup \set{v\in V\mid \phi_v+\psi_v\geq 1}.
    \]
\end{enumerate}

\begin{algorithm}[H]
\caption{\textsc{WaterFilling}($H = (V_H,E_H)$)}
\begin{algorithmic}[1]
\State For all edges $e$, $\phi_e\leftarrow 0$, initially
\State Call a vertex $v$ \emph{inactive} if $\phi_v\coloneqq \sum_{e\ni v}\phi_e=1$ and \emph{active} otherwise; similarly, an edge is {active} iff both its endpoints are active and inactive otherwise 
\While{{there exists an active edge} and $\phi_e<\varepsilon_1p_e$ for all edges $e$}
\State Pick minimum $\delta\in (0,1)$ such that $\phi_{e}\leftarrow \phi_e+\delta p_e$ for all $e$ active results in at least one new inactive vertex or $\phi_e=\varepsilon_1p_e$ for some edge $e$
\State $\phi_e\leftarrow \phi_e+\delta p_e$ for all active edges $e$
\EndWhile
\end{algorithmic}
\label{alg:waterfilling}
\end{algorithm}

\begin{algorithm}[H]
\caption{\textsc{DistributedWaterFilling}($H=(V_H,E_H), b=(b_v)_{v\in V_H}$)}
\begin{algorithmic}[1]
\State For all edges $e$, $\psi_e\leftarrow 0$, initially
\State Call a vertex $v$ \emph{inactive} if $\psi_v\coloneqq \sum_{e\ni v}\psi_e \geq 1- b_v$ 
and \emph{active} otherwise; similarly, an edge is {active} iff both its endpoints are active and inactive otherwise
\While{{there exists an active edge} }
\State $\psi_e\leftarrow \psi_e+\frac{\varepsilon_3}{\xi}$ for all active edges $e$ \label{line:increaseweight}
\EndWhile
\end{algorithmic}
\label{alg:distributedwaterfilling}
\end{algorithm}

We now proceed to the analysis of the algorithm.

\begin{lemma}
    $C$ is a valid vertex cover, and the number of rounds of the distributed phase is $O(\xi\cdot\varepsilon_3^{-1})=O(1/\varepsilon^4)$.
\end{lemma}
\begin{proof}
    Since in \Cref{alg:distributedwaterfilling} Line \ref{line:increaseweight}, we increase the weight on the edge by $\frac{\varepsilon_3}{\xi}$ in each round, the total number of rounds until a vertex becomes inactive is at most $O(\xi\cdot \varepsilon_3^{-1})$. To see that $C$ is a valid vertex cover, first observe that any edges with an end point in $F\cup B$ are automatically covered by $C$. For any $e=\{u,v\}\in Q^\ast$, \Cref{alg:distributedwaterfilling} ensures that either $\phi_v+\psi_v\geq1$ or $\phi_u+\psi_u\geq1$. Thus, by definition of $C$, at least one of $u$ or $v$ are included in $C$.
\end{proof}

In the remaining of this section we bound the approximation ratio of the algorithm.

\begin{lemma}
    \label{lem:keren}
    For all $e\in E$, $\phi_e \leq \varepsilon_1\cdot p_e$, and, in particular, $\phi_e = \varepsilon_1\cdot p_e$ if $e\cap F = \emptyset$ (i.e., $e$ is still active at the end of \Cref{alg:waterfilling}).   
\end{lemma}
\begin{proof}
    Let $\delta_i$ denote the increment $\delta$ that the algorithm uses in the $i$th while-loop iteration, and let $t$ denote the total number of iterations.
    At the end of the loop, we have for every edge $e$,
    \[
        \phi_e \leq \sum_{i=1}^t \delta_i\cdot p_e.
    \]
    In particular, if $e$ is still active then the above inequality holds as equality.
    If at the end of the algorithm there is some active edge then by the exit condition of the while loop, and by the way $\delta$ is chosen it follows that there is some active edge $e$ satisfying
    \[
        \phi_e = \varepsilon_1\cdot p_e.
    \]
    It follows that $\sum_{i=1}^t \delta_i = \varepsilon_1$, in this case.
    If there is no active edge at the end of the loop, then we similarly conclude that $\sum_{i=1}^t \delta_i \leq \varepsilon_1$.
    In both cases the lemma follows.
\end{proof}

\begin{lemma}
    \label{lem:mxdgrGQast}
    The maximum degree of subgraph $(V,Q^\ast)$ is at most $\xi = \frac{1+\varepsilon_2}{\varepsilon_1}$.
\end{lemma} 
\begin{proof}
    If $v\in F\cup B$ then $v$'s degree in $(V,Q^\ast)$ is zero.
    Let $v\in V\setminus(F\cup B)$.
    Let $d$ be $v$'s degree in $(V,Q^\ast)$, and let $e_1,\ldots,e_d$ be its incident edges.
    For each $1\leq i\leq d$, we have $e_i \cap F = \emptyset$, thus \cref{lem:keren} implies
    \[
        \phi_{e_i} = \varepsilon_1\cdot p_{e_i}
        .
    \]
    Then
    \[
        \chi_v 
        = 
        \sum_{e \ni v}\chi_e
        \geq 
        \sum_{i=1}^d \chi_{e_i}
        =
        \sum_{i=1}^d \frac{\phi_{e_i}}{p_{e_i}}
        =
        d\cdot \varepsilon_1
        .
    \]
    Also, since $v\notin B$ we have
    \[
        \chi_v < \phi_v + \varepsilon_2 \leq 1 + \varepsilon_2,
    \]
    where the second inequality holds because $\phi$ is valid fractional matching of $G$.
    Combining the last two equation gives $d\leq \frac{1+\varepsilon_2}{\varepsilon_1}$.
\end{proof}

The statement of the next lemma is similar to \cite[Lemma 6.2]{BehnezhadBD22}, but several elements of the proof are different.

\begin{lemma}
    \label{lem:FMy}
    There is a fractional matching $y=(y_e)_{e\in E^\ast}$ such that 
    \begin{equation}
        \label{eq:expCub}
        \Exp[|C|] \leq (2+\epsilon)\Exp[|y|]
        ,
    \end{equation}
    where $\epsilon\leq 10\varepsilon$.
\end{lemma}

We start with defining our \emph{candidate} fractional matching $y$.  Let $\varepsilon_4 = \varepsilon_2 + \varepsilon_3$. For each $e\in E^\ast$, 
    \[  
        y_e = \left(\chi_e + \psi_e\right) \cdot \frac1{1+\varepsilon_4} \cdot \mathds{1}_{e\cap B = \emptyset}.
    \]
That is, $y_e = 0$ if $e$ has an endpoint in $B$ and is $\frac{\chi_e+ \psi_e}{1+\varepsilon_4}$ otherwise. We also define $y_e = 0$ if $e\in E\setminus E^\ast$.

\begin{lemma}
    The candidate fractional matching $y$ is a valid fractional matching.
\end{lemma}
\begin{proof}
    For all $v\in B$, we have $y_v=0$, thus satisfying the fractional matching constraints trivially. 
    Let $v\notin B$. Then
    \[
        y_v\leq 
        \sum_{e\ni v}
        \frac{\chi_e+\psi_e}{1+\varepsilon_4}
        \leq \frac{\chi_v+\psi_v}{1+\varepsilon_4}
        .
    \]
    From \cref{alg:distributedwaterfilling}, we have that $\psi_v \leq 1-\phi_v$ while $v$ is still active, and if $v$ becomes inactive, then in the iteration when $v$ first becomes inactive, $\psi_v$ increases by at most $\frac{\varepsilon_3}{\xi}$ times the degree of $v$ in $(V,Q^\ast)$, which is at most $\xi$, from \cref{lem:mxdgrGQast}.
    Thus, at the end of \cref{alg:distributedwaterfilling},
    \[
        \psi_v \leq 1-\phi_v + \frac{\varepsilon_3}{\xi} \cdot \xi
        =
        1-\phi_v+\varepsilon_3
        .
    \]
    Also $v\notin B$ implies 
    $
        \chi_v 
        \leq \phi_v+\varepsilon_2
    $.
    Substituting these two inequalities above gives
    \[
        y_v
        \leq 
        \frac{1+\varepsilon_2+\varepsilon_3}{1+\varepsilon_4}
        =
        1
        ,
    \]
    thus $v\in B$ satisfies the fractional matching constraint, as well.
\end{proof}

We define the following set $B^+\supseteq B$, which uses a slightly lower threshold value than $B$, namely $\phi_v + \varepsilon_2 - \varepsilon_1$ instead of $\phi_v + \varepsilon_2$,
    \[
        B^+ = \left\{v\in V \colon  \chi_v \geq \phi_v + \varepsilon_2 - \varepsilon_1  \right\}.
    \]
    
The next lemma establishes a lower bound on $\Exp[|y|]$.
\begin{lemma}
    \label{lem:expy}
    $
    (1+\varepsilon_2)\cdot \Exp[|y|]
    \geq
    |\phi|
    +
    \Exp[|\psi|]
    -
    \Exp[|B^+|]
    .
    $
\end{lemma}
\begin{proof} 
    We have
    \begin{align*}
        (1+\varepsilon_4)\cdot \Exp[|y|]
        &=
        (1+\varepsilon_4) \sum_{e\in E}\Exp[y_e]
        \\&
        =
        \sum_{e\in E}\Exp\left[
        \left(\chi_e + \psi_e\right) \cdot \mathds{1}_{e\cap B = \emptyset}
        \cdot \mathds{1}_{e\in E^\ast}
        \right]
        \\&
        =
        \sum_{e\in E}\chi_e \cdot \Pr(
         e\cap B = \emptyset \,\land\, 
        e\in E^\ast
        )
        +
        \sum_{e\in E}\Exp\left[
        \psi_e  \cdot \mathds{1}_{e\cap B = \emptyset}
        \cdot \mathds{1}_{e\in E^\ast}
        \right]
        \\&
        =
        \sum_{e\in E}\chi_e\cdot p_e\cdot \Pr(
        e\cap B = \emptyset \mid e\in E^\ast
        )
        +
        \sum_{e\in E}\Exp\left[
        \psi_e  \right]
        \\&
        =
        \sum_{e\in E}\phi_e\cdot (1-\Pr(
        e\cap B \neq \emptyset \mid e\in E^\ast)
        )
        +
        \Exp[|
        \psi| ]
        \\&
        =
        |\phi| - \sum_{e\in E}\phi_e\cdot \Pr(
        e\cap B \neq \emptyset \mid e\in E^\ast)
        +
        \Exp[|
        \psi| ].
    \end{align*}
    Next we will use that if $e\ni v$ then $\Pr(v\in B \mid e\in E^\ast) \leq \Pr(v\in B^+)$, because by \cref{lem:keren}, a single edge $e$ can contribute at most $\varepsilon_1$ to $\chi_v$, 
    in the definition of $B$ and $B^+$.
    We bound the sum in the last line above as follows, 
    \begin{align*}
        \sum_{e\in E} \phi_e\cdot \Pr(
        e\cap B \neq \emptyset \mid e\in E^\ast)
        &\leq
        \sum_{v\in V}
        \sum_{e\ni v}
        \phi_e
        \cdot
        \Pr(v\in B \mid e\in E^\ast)
        \\&
        \leq
        \sum_{v\in V}
        \sum_{e\ni v}
        \phi_e
        \cdot
        \Pr(v\in B^+)
        \\&
        =
        \sum_{v\in V}
        \Pr(v\in B^+)
        \cdot \sum_{e\ni v}
        \phi_e
        \\&
        \leq
        \sum_{v\in V}
        \Pr(v\in B^+)
        \cdot 1
        \\&
        \leq
        \Exp[|B^+|]
        .
    \end{align*}
    Substituting this above completes the proof.
\end{proof}

We now give an upper bound on $\Exp[|C|]$.
\begin{lemma}
    \label{lem:expC}
    $
    \Exp[|C|]
    \leq
    2|\phi|
    +
    2\Exp[|\psi|]
    +
    \Exp[|B|]
    .
    $
\end{lemma}
\begin{proof}
    Recall $C = F\cup B\cup \{v\in V \colon \phi_v + \psi_v \geq 1\}$.
    We have
    \begin{align*}
        \Exp[|C|]
        &=
        |F| + \Exp[|B\setminus F|] + \Exp[|C\setminus (F\cup B)|]
        \\&
        \leq
        |F| + \Exp[|B|] + \Exp[|C\setminus F|]
        .
    \end{align*}
    Also, from the definition of $C$ it follows
    \begin{align*}
        \Exp[|C\setminus F|]
        &\leq
        \sum_{v\in V\setminus F}
        \Exp[\phi_v + \psi_v ]
        \\&
        =
        \sum_{v\in V}
        \phi_v
        -
        \sum_{v\in F}
        \phi_v
        +
        \sum_{v\in V}
        \Exp[\psi_v ]
        \\&
        =
        2|\phi| - |F| + 2\Exp[|\psi|]
        .
    \end{align*}
    Substituting this above completes the proof.
\end{proof}

Finally, we show an upper bound on $\Exp[|B^+|]$.
Recall $\varepsilon_5 = \frac{\varepsilon_1}{(\varepsilon_2 - \varepsilon_1)^2}$.

\begin{lemma}
    \label{lem:expB}
    $\Exp[|B|] \leq \Exp[|B^+|] \leq \varepsilon_5 \cdot |\phi|$.
\end{lemma}
\begin{proof}
    The first inequality is immediate from the definition of the two sets.
    We prove the second inequality.
    Consider any vertex $v\in V$.
    Let $d$ be the degree of $v$ in the base graph $G$, and let $e_1,\ldots,e_d$ be its incident edges.
    For $1\leq i\leq d$, 
    \[
        \Exp[\chi_{e_i}]
        =
        \frac{\phi_{e_i}}{p_{e_i}}\cdot p_{e_i} 
        =
        \phi_{e_i}
        ,
    \]
    and
    \[
        \Var[\chi_{e_i}] 
        = 
        \left(\frac{\phi_{e_i}}{p_{e_i}}\right)^2p_{e_i}(1-p_{e_i}) 
        \leq \frac{\phi_{e_i}^2}{p_{e_i}}
        \leq
        \varepsilon_1 \phi_{e_i}
        ,
    \]
    where the last inequality is obtained using \cref{lem:keren}.
    It follows that 
    $
        \Exp[\chi_v] = \phi_v, 
    $    
    and 
    \[
        \Var[\chi_v] = \sum_{i=1}^d \Var[\chi_{e_i}] \leq \sum_{i=1}^d \varepsilon_1 \phi_{e_i}
        = \varepsilon_1 \phi_{v}
        .
    \]
    Then the probability that vertex $v$ is in $B^+$ is 
    \begin{align*}
        \Pr(v\in B^+)
        &=
        \Pr(\chi_v > \phi_v + \varepsilon_2 - \varepsilon_1)
        \\&
        \leq
        \Pr(|\chi_v - \phi_v| > \varepsilon_2 - \varepsilon_1)
        \\&
        \leq
        \frac{\Var[\chi_v]}{(\varepsilon_2 - \varepsilon_1)^2}
        \\&
        \leq
        \frac{\varepsilon_1\phi_v}{(\varepsilon_2 - \varepsilon_1)^2}
        .
    \end{align*}
    Finally,
    \[
        \Exp[|B^+|]
        =
        \sum_{v\in V} \Pr(v\in B^+)
        \leq
        \frac{\varepsilon_1\cdot |\phi|}{(\varepsilon_2 - \varepsilon_1)^2} 
        .
        \qedhere
    \]
\end{proof}

From \cref{lem:expy,lem:expB},
\begin{align*}
    (1+\varepsilon_4)\cdot \Exp[|y|]
    &\geq
    |\phi|
    +
    \Exp[|\psi|]
    -
    \Exp[|B^+|] 
    \\&
    \geq
    (1 - \varepsilon_5)\cdot |\phi|
    +
    \Exp[|\psi|]
    \\&
    \geq
    (1 - \varepsilon_5)\cdot (|\phi|
    +
    \Exp[|\psi|])
    .
\end{align*}    
Similarly, from \cref{lem:expC,lem:expB},
\begin{align*}
    \Exp[|C|]
    &\leq
    2|\phi|
    +
    2\Exp[|\psi|]
    +
    \Exp[|B|]
    \\&
    \leq
    \left(2 + \varepsilon_5\right)\cdot (|\phi|
    +
    \Exp[|\psi|]).
\end{align*}
Combining the two inequalities yields
\[
    \Exp[|C|] \leq \left(2 + \varepsilon_5\right) \cdot \frac{1+\varepsilon_4}{1 - \varepsilon_5}\cdot \Exp[|y|]
    ,
\]
which implies \cref{eq:expCub}, for 
$\epsilon = \frac{(2 + \varepsilon_5)(1+\varepsilon_4)}{1 - \varepsilon_5}-2$.
This completes the proof of \cref{lem:FMy}.
From \cref{lem:FMy}, it follows that $C$ is a $(2+\epsilon)$-approximation of the optimal matching of $G^\ast$ as argued in \cite{BehnezhadBD22}, since any fractional matching is smaller than the minimum vertex cover.

\section{Constant Approximation for Maximum Matching in 2 Rounds}
\label{sec:mmTwoRounds}

We give an algorithm that finds a matching of $G^\ast$ with the following properties.

\constMatchingTwoRounds*

The algorithm takes a matching algorithm proposed in \cite{BehnezhadBD22} for bipartite graphs and extends it to general graphs.
The approximation ratio of that algorithm is $1-1/e=0.63212$. 

Let $\MM$ be any matching algorithm, and let $\MM(H)$ denote the output of the algorithm on input graph $H$. 
For an edge $e=uv$, we use notation $e-u$ to denote $v$. 

\paragraph{Distributed Algorithm} 
\begin{enumerate}
    \item ({\bf Preprocessing Stage}) The following two tasks are performed by each $v\in V$. (The value of $\alpha$ will be picked suitably later.)
    \begin{enumerate}
        \item  Vertex $v$ chooses to be \emph{active} with probability $\alpha$, or \emph{passive} with probability $1-\alpha$. We denote the set of active vertices by $A$ and the set of passive vertices by $P$. 

    \item Each $v\in V$ communicates its choice to its neighbors in $G$, and thus it learns its set of active and passive neighbors in $G$.
    \end{enumerate}
    \item ({\bf First Round}) At this stage, each $v\in V$ is made aware of $N_{G^\ast}(v)$. If $v\in A$, then it may `propose' to at most one neighbor:
    \begin{enumerate}
    
    \item Each vertex $v\in A$ generates a random hallucination 
    
    $G_v^\ast$ of $G$ as follows.
    For every edge $vu\in E(G)$ where $u$ is passive, $vu \in E(G^\ast_v)$ if and only if $vu \in E(G^\ast)$; for every other $e \in E(G)$, edge $e$ is added to $E(G^\ast_v)$ independently with probability $p_{e}$. 
    
    \item \label{item:2b}
    Next each $v\in A$ locally computes $\MM(G^{\ast}_{v})$, and if $vu \in \MM(G^{\ast}_{v})$ for some $u\in N_{G^\ast}(v)$,
    then $v$ sends a proposal to $u$.
    
    \end{enumerate}
    
    \item ({\bf Second Round}) In this round, vertices accept or reject proposals they receive:
    \begin{enumerate}
    \item If $v\in A$, then $v$ rejects all proposals it receives.
    \item If $v\in P$    
    and $v$ received at least one proposal in the first round, then it accepts an arbitrary one among the proposals, and notifies the sender $u$. Edge $vu$ is added to the matching $M$ output by the algorithm (and is reported by $v$ and $u$).   
    \end{enumerate}
\end{enumerate}

The algorithm clearly outputs a matching since each $v\in A$ makes a single proposal and rejects all proposals made to it; and each each $v\in P$ accepts an arbitrary proposal.
Next we analyze the approximation ratio.

Although in the algorithm only active vertices $v\in A$ hallucinate a graph $G^\ast_v$, for the analysis we will assume that passive vertices do that as well. 
Precisely, we assume each $v\in P$ generates a random hallucination $G^{\ast}_{v}$ by adding each edge $e\in E(G)$ to $G_{v}^{\ast}$ independently with probability $p_e$.

We define now some relevant random variables associated with the algorithm.

\begin{definition}[Random Variables] We will use the following random variables:
    \begin{enumerate}
        \item For each $u\in V$ let $S_u$ be the indicator random variable that is 1 if $u\in A$ and 0 if $u\in P$. 
        \item For each $e\in E(G)$ let $R_e$ be the random variable that is 1 if $e$ is realized in $G^\ast$ and 0 otherwise.
        \item For each $e\in E(G)$ and $v\in V$ let $R^v_{e}$ be the random variable that is 1 if $e$ is realized in $G^\ast_v$ and 0 otherwise.
        
    \end{enumerate}
\end{definition}



\begin{lemma}
    \label{lem:indep}
    All hallucinations $G_v^\ast$, for $v\in V$, are mutually independent and also independent of the set $A$ of active vertices. 
\end{lemma}
\begin{proof}
    For $v\in A$, consider $G^\ast_{v}$. The edges $vw\in E(G)$ are of two types: $w\in P$ or $w\in A$. 
    In the former case, $vw$ is realized consistently with its status in $G^\ast$. However, $G^\ast$ is independent of $A$ and $P$. 
    Similarly, in the latter case, $vw$ is realized independently with probability $p_{vw}$. 
    Thus, for all such $v\in A$, $G^\ast_{v}$ is independent of $A$. The argument for $v\in P$ is similar. In this case, all edges $vw\in G$ are realized independently with probability $p_{vw}$. Thus for all such $v\in P$, $G_{v}^{\ast}$ is independent of $A$.

    To argue about mutual independence of all $G^\ast_v$ for $v\in V$, we observe the following:
    \begin{itemize}
        \item For $v\in A$, $E(G^\ast_v)$ is determined by $\set{R_e}_{e\in N_{G^\ast}(v),e-v\in P}$ and $\set{R^v_e}_{e\in N_{G^\ast}(v),e-v\in A}$.
        \item For $v\in P$, $E(G^\ast_v)$ is determined by $\set{R^v_{e}}_{e\in N_{G^\ast}(v)}$. 
    \end{itemize}
Thus, each $v\in V$ uses disjoint random bits to sample $G^\ast_v$. Therefore, we can conclude that all $G^\ast_v$ are mutually independent. 
\end{proof}

\begin{lemma}
\label{lem:matching-const}
    The expected size of the matching $M$ computed by the distributed algorithm is 
    \[
        \Exp[|M|] \geq 2(1-\alpha) \cdot(1-e^{-\alpha})\cdot \Exp[|\MM(G^\ast)|]
        .
    \]
\end{lemma}

\begin{proof}
    For every vertex $v\in V$, let $\EE_v$ be the event that $v$ is passive and receives at least one proposal.
    Then  
    \[
        \Exp[|M|] = \sum_{v\in V} \Pr(\EE_v).
    \]
   
    For every pair $u,v\in V$, let $X_{u, v}$ be the indicator random variable of the event $u\in A\,\land\, uv \in\MM(G^\ast_u)$. 
    Note that the event $\EE_v$ is the same as the event 
    \[
        v\in P \,\land\, \sum_{u\in N_G(v)} X_{u,v}>0.
    \]

    We have that
    \begin{align*}
        \Pr(X_{u, v} = 0)
        &=
        1 - \Pr(X_{u, v} = 1)
        \\&
        =
        1 - \Pr(u\in A \land uv \in \MM(G_u^\ast))
        \\&
        =
        1 - \Pr(u\in A)\cdot \Pr(uv \in \MM(G_u^\ast))
        \\&
        =
        1 - \alpha\cdot \Pr(uv \in \MM(G^\ast))
        \\&
        \leq
        e^{- \alpha\cdot \Pr(uv \in \MM(G^\ast))}
        ,
    \end{align*}
    where in the third line we used that $G_u^\ast$ is independent of $A$ by \cref{lem:indep},
    and in the fourth line we used that $G^\ast$ and $G_u^\ast$ have the same distribution.
    Moreover,  since $X_{u, v}$ depends only on $S_u$, $G^\ast_u$,  \cref{lem:indep} implies that the random variables $X_{u, v}$, for $u\in N_G(v)$, are mutually independent, and also independent of $S_v$. 
    Then for every $v\in V$,
    \begin{align*}
        \Pr(\EE_v)
        &=
        \Pr\!\left(v\in P \land \sum_{u\in N_G(v)} X_{u, v} > 0\right)
        \\
        &=\Pr\!\left(v\in P\right)\cdot \Pr\!\left(\sum_{u\in N_G(v)}X_{u, v}>0\right)
        \\
        &=
        (1-\alpha)\cdot \left(1-\Pr\!\left(\sum_{u\in N_G(v)} X_{u, v} = 0\right)\right)
        \\
        &=
        (1-\alpha)\cdot \left(1-\prod_{u\in N_G(v)}\Pr(X_{u, v} = 0)\right)
        \\
        &\geq
        (1-\alpha)\cdot \left(1-e^{-\alpha\cdot \sum_{u\in N_G(v)}\Pr(uv \in \MM(G^\ast))}\right)
        ,
    \end{align*}
    where in the third line we used that the $X_{u, v}$ are independent of $S_v$; in the second-last line we used that the $X_{u, v}$, for $u\in N_G(v)$, are mutually independent; and in the last line we applied the previous inequality.
    The sum in the exponent in the last line above equals the probability $c_v$ that $v$ is matched in $\MM(G^\ast)$.Thus,
    \[
        \Pr(\EE_v) \geq
        (1-\alpha)\cdot \left(1-e^{-\alpha c_v}\right)
        =
        (1-\alpha)\cdot c_v \cdot \frac{1-e^{-\alpha c_v}}{c_v}
        \geq
        (1-\alpha)\cdot c_v \cdot (1-e^{-\alpha}),
    \]
    since the function $\frac{1-e^{-\alpha x}}{x}$ is decreasing and $c_v \leq 1$. 
    Thus we have,
    \[
        \Exp[|M|] 
        = 
        \sum_{v\in V} \Pr(\EE_v)
        \geq
        (1-\alpha) \cdot (1-e^{-\alpha})\cdot \sum_{v\in V}c_v
        =
        2(1-\alpha) \cdot (1-e^{-\alpha})\cdot \Exp[|\MM(G^\ast)|]
        .
        \qedhere
    \]
\end{proof}

For $\alpha = 1/2$, \cref{lem:matching-const} gives an approximation ratio of
\[
    1-e^{-1/2} \approx 0.393469.
\]
The optimal choice for $\alpha$ is $\alpha = 0.442854$, which gives a slightly better approximation ratio
\[
    2(1-\alpha)(1-e^{-\alpha}) \approx 0.398693
    ,
\]
which proves \Cref{thm:constMatchingTwoRounds}.

\section{Improved Approximation for Maximum Matching in \texorpdfstring{$\poly(1/\eps)$}{poly(1/eps)} Rounds}
\label{sec:matching-const-approx-polyeps}

We give a $\poly({1}/{\epsilon})$-round algorithm that computes a $(0.68-\epsilon)$-approximation when $G$ is an arbitrary graph. For the case of bipartite graphs, we give a better $\poly({1}/{\epsilon})$-round algorithm that computes a $(0.73-\epsilon)$-approximation.  

\MatchingInPolyRrounds*

We will use an existing result on distributed algorithms for $(1-\delta)$-approximate maximum matching. 

\begin{lemma}[\cite{Bar-YehudaCGS17}]\label{lem:localapproxmatching}
    There is a \congest algorithm $\mathcal{B}$ that computes a $(1-\delta)$-approximation to the maximum matching on any given graph $G$ with maximum degree $\Delta$, in $O(\poly({1}/{\delta})\cdot \frac{\log \Delta}{\log\log \Delta})$ rounds. 
\end{lemma}

We will also use the following statement from \cite{DerakhshanS25}. 

\begin{lemma}[{\cite[Theorem 1]{DerakhshanS25}}]\label{lem:DS25}
 Let $\MM$ be any (potentially randomized) maximum matching algorithm. There exists an algorithm $\mathcal{A}$ that picks a subgraph $Q$ of $G$ with maximum degree $O(\frac{1}{\epsilon^5p})$ such that,
$\Exp[|\MM(Q^\ast)|]\geq (0.68-\epsilon)\cdot \Exp[|\MM(G^\ast)|]$. For the case when $G$ is bipartite, we can get an improved algorithm $\mathcal{A}'$ that picks a subgraph $Q$ of $G$ with maximum degree $O(\frac{1}{\epsilon^5p})$ such that $\mathds{E}[|\MM(Q^*)|]\geq (0.73-\epsilon)\cdot \mathds{E}[|\MM(G^*)|]$.
\end{lemma}

We now describe our algorithm $\alg$ which obtains a $(0.68-\epsilon)$-approximation to $\Exp[|\MM(G^\ast)|]$ in $\poly({1}/{\epsilon})$ rounds when $G$ is a general graph. Additionally, $\alg$ obtains a $(0.73-\eps)$-approximation in $\poly({1}/{\epsilon})$ when $G$ is bipartite. 
\begin{enumerate}
    \item ({\bf Preprocessing Stage}) We use $\mathcal{A}$ from \Cref{lem:DS25} to get a subgraph $Q$ of $G$ satisfying the statement of the lemma if $\mathcal{A}$ is non-bipartite. Otherwise, we use $\mathcal{A}'$ to pick $Q$. 
    \item ({\bf Distributed Stage}) We consider $Q^\ast$, drop all vertices $V_{\text{bad}}=\set{v\in V\mid \text{deg}_{Q^\ast}(v)\geq \frac{1}{\epsilon^{10}}}$ and proceed to run $\mathcal{B}$ (guaranteed by \Cref{lem:localapproxmatching}) with $\delta=\frac{\epsilon}{2}$ on $Q^\ast[V\setminus V_{\text{bad}}]$. Then, we output the matching $M$ computed by $\mathcal{B}$.
\end{enumerate}

Note that $\Exp[|\MM(Q^\ast)|]$ is large compared to $\Exp[|\MM(G^\ast)|]$, but we compute a matching in graph $Q^\ast[V\setminus V_{\text{bad}}]$ to ensure a round complexity that is independent of $p$. It remains to argue that $\Exp[|\MM(Q^\ast[V\setminus V_{\text{bad}}])|]\geq (1-\epsilon)\cdot \Exp[|\MM(Q^\ast)|]$. This will be the focus of the subsequent lemma. Let
$\theta=\frac{1}{\epsilon^{10}}$. 

\begin{lemma}\label{lem:boundingbadvertices1}
     Let $\epsilon<\frac{1}{2}$ and let $Q$ be a subgraph of $G$ such that maximum degree of $Q$ is upper bounded by $O(\frac
   {1}{\epsilon^5p})$. Suppose 
    \[
        V_{\text{bad}}\coloneqq\set{v\in V\colon \textup{deg}_{Q^\ast}(v)\geq \frac{1}{\epsilon^{10}}}.
    \]
    Let $X_v\coloneqq \textup{deg}_{Q^\ast}(v)$, for $v\in V$. Then $\Exp\left[|Q^\ast[V\setminus V_{\textup{bad}}]|\right]\geq \sum_{v\in V}\frac{\Exp[X_v]}{2}\cdot\left(1-\frac{8}{\theta}\right)$.
\end{lemma}
\begin{proof}
     Note that $\Exp[|Q^\ast|]=\sum_{v\in V} \frac{\Exp[X_v]}{2}$ and $\Exp[X_v]=\frac{1}{\epsilon^8}$ (since by \Cref{lem:DS25}, $\text{deg}_{Q}(v)\leq \frac{1}{\epsilon^8p}$).  We can conclude that
    \begin{align*}
        \Pr(X_v\geq \Exp[X_v]+k)\leq \frac{\Var[X_v]}{k^2}\leq \frac{\Exp[X_v]}{k^2},
    \end{align*}
where the last inequality is derived from the fact that $X_v$ is a sum of 0-1 independent random variables. With this in mind, we 
proceed to bound the term $\Exp[X_v\cdot\mathds{1}_{X_v\geq \theta}]$. This term is an upper bound on the expected number of edges incident on $V_{\text{bad}}$. Consequently, upper bounding this term will let us lower bound $\Exp\!\left[|Q^\ast[V\setminus V_{\text{bad}}]|\right]$. We have,
\begin{align*}
    \Exp[X_v\cdot \mathds{1}_{X_v\geq \theta}]&=\sum_{j\geq \theta }\Pr(X_v\geq j)\\
    &\leq \sum_{j\geq \theta} \frac{\Exp[X_v]}{(j-\Exp[X_v])^2}\\
    &\text{(by the discussion above)}\\
    &\leq \sum_{j\geq \theta} \frac{\Exp[X_v]}{(j-\epsilon^2\theta)^2}\\
    &\text{(since $\Exp[X_v]\leq \epsilon^2\theta$)}\\
    &\leq \Exp[X_v]\cdot \sum_{i\geq 0}\sum_{j\in [2^i\theta,2^{i+1}\theta]}\frac{1}{(2^{i}-\epsilon^2)^2\theta^2}\\
    &\leq \Exp[X_v]\cdot \sum_{i\geq 0} \frac{2^{i}\theta}{(2^i-\epsilon^2)^2\theta^2}\\
    &\leq \Exp[X_v]\cdot \sum_{i\geq 0} \frac{2^{i}}{(2^{i-1})^2\theta}\\
    &\text{(as $\eps<\frac{1}{2}$)}\\
    &\leq \frac{4\cdot \Exp[X_v]}{\theta}.
\end{align*}
Thus, we can conclude that
\begin{align*}
    \Exp[|Q^\ast[V\setminus V_{\text{bad}}]|]\geq \Exp[|Q^\ast|]-\sum_{v\in V}\Exp[X_v\cdot \mathds{1}_{v\geq \theta}]\geq \sum_{v\in V}\frac{\Exp[X_v]}{2}\cdot\left(1-\frac{8}{\theta}\right).&\qedhere  
\end{align*}
\end{proof}

Our second lemma gives a lower bound $\Exp[|\MM(Q^\ast)|]$.
\begin{lemma}\label{lem:lowerboundmatchingsize}
Let $\mathcal{M}$ be an maximum matching algorithm. Let $\epsilon<\frac{1}{2}$ and let $Q$ be a subgraph of $G$ such that maximum degree of $Q$ is upper bounded by $O(\frac
   {1}{\epsilon^5p})$. Suppose 
    \[
        V_{\text{bad}}\coloneqq\set{v\in V\colon \textup{deg}_{Q^\ast}(v)\geq \frac{1}{\epsilon^{10}}}.
    \] 
Then we have that $\Exp[|\MM(Q^\ast)|]\geq \Exp[|\MM(Q^\ast[V\setminus V_{\textup{bad}}])|]\geq \sum_{v\in V}\frac{\Exp[X_v]\cdot (1-\frac{8}{\theta})}{2\theta}$.
\end{lemma}
\begin{proof}
    We create a matching $M$ of $Q^\ast[V\setminus V_{\text{bad}}]$ with expected size being equal to the RHS as follows. Initially, $M=\emptyset$ and $H=Q^\ast[V\setminus V_{\text{bad}}]$. In each step, we consider an arbitrary edge $e=uv\in H$ and let $M\leftarrow M\cup \{e\}$. Additionally, we update $H$ by removing from $H$ the vertices $u,v$ and all the edges incident on them. We finalize $M$ when $H=\emptyset$. Since $v\in V\setminus V_{\text{bad}}$ has $\text{deg}_{Q^\ast}(v)\leq \theta$, $H=\emptyset$ after at least $\sum_{v\in V}\frac{\Exp[X_v]\cdot (1-\frac{8}{\theta})}{2\theta}$ steps in expectation. Thus, we have, 
    \begin{align*}
        \Exp[|\MM(Q^\ast)|]\geq \Exp[|\MM(Q^\ast[V\setminus V_{\text{bad}}])|]\geq \Exp[|M|]\geq \sum_{v\in V}\frac{\Exp[X_v]\cdot (1-\frac{8}{\theta})}{2\theta}.&\qedhere
    \end{align*}
\end{proof}

Our next lemma bounds $|V_{\text{bad}}|$.
\begin{lemma}\label{lem:upperboundVbad}
    We have $\Exp[|V_{\textup{bad}}|]\leq \frac{16\cdot\sum_{v\in V}\Exp[X_v]}{9\theta^2}\leq \left(\frac{32\cdot \Exp[|\MM(Q^\ast)|]}{9\theta}\right)\cdot\left(\frac{1}{1-\frac{8}{\theta}}\right)$.
\end{lemma}
\begin{proof}
    We have
    \begin{align*}
        \Pr(X_v\geq \theta)\leq \frac{\Var[X_v]}{(\theta-\Exp[X_v])^2}\leq \frac{\Exp[X_v]}{(\theta-\epsilon^2\theta)^2}\leq \frac{16\cdot \Exp[X_v]}{9\theta^2}.
    \end{align*}
The last inequality follows from the fact that $\eps<\frac{1}{2}$. Thus, we have,
\begin{align*}
    \Exp[|V_{\text{bad}}|]=\sum_{v\in V}\Exp[\mathds{1}_{X_v\geq \theta}]\leq \frac{16\cdot\sum_{v\in V}\Exp[X_v]}{9\theta^2}\leq \left(\frac{32\cdot \Exp[|\MM(Q^\ast)|]}{9\theta}\right)\cdot\left(\frac{1}{1-\frac{8}{\theta}}\right)
    ,
\end{align*}
where the last inequality follows from \Cref{lem:lowerboundmatchingsize}.
\end{proof}

We can conclude the following.

\begin{lemma}\label{lem:boundingbadvertices}
Let $\MM$ be any maximum matching algorithm.   Let $\epsilon<\frac{1}{2}$ and let $Q$ be a subgraph of $G$ such that maximum degree of $Q$ is upper bounded by $O(\frac
   {1}{\eps^5p})$. Suppose 
    \[
        V_{\text{bad}}\coloneqq\set{v\in V\colon \textup{deg}_{Q^\ast}(v)\geq \frac{1}{\eps^{10}}}.
    \]
Then, $\Exp[|\MM(Q^\ast[V\setminus V_{\textup{bad}}])|]\geq (1-\eps^7)\cdot \Exp[|\MM(Q^\ast)|]$
\end{lemma}
\begin{proof}
As a consequence of \Cref{lem:upperboundVbad}, we have that 
\begin{align*}
    \Exp[|\MM(Q^\ast[V\setminus V_{\textup{bad}}])|]&\geq \left(1-\frac{32}{9\theta}\cdot\left(\frac{1}{1-\frac{8}{\theta}}\right)\right)\cdot \Exp[|\MM(Q^\ast)|]\\
    &\geq \left(1-3.56\eps^{10}\cdot\left(1+16\eps^{10}\right)\right)\cdot \Exp[|\MM(Q^\ast)|]\\
    &\text{(By definition of $\theta$)}\\
    &\geq (1-8\eps^{10})\cdot\Exp[|\MM(Q^\ast)|]\\
    &\geq (1-\eps^{7})\cdot \Exp[|\MM(Q^\ast)|]\\
    &\text{(As $\eps<\frac{1}{2}$).}\qedhere
\end{align*}
\end{proof}
The above completes the proof of our algorithm, as follows.
\begin{proof}[Proof of \Cref{thm:MatchingInPolyRrounds} (General Graphs)]
Let $\alg$ be the candidate algorithm. If $G$ is a non-bipartite graph $\alg$ uses $\mathcal{A}$ to pick $Q$, which satisfies the premise of \Cref{lem:boundingbadvertices}. Thus, we have that 
\begin{align*}
    \Exp[|\MM(Q^\ast[V\setminus V_{\textup{bad}}])|]
    &\geq (1-\eps^{7})\cdot \Exp[|\MM(Q^\ast)|]
\end{align*}
Since we run $\mathcal{B}$ on $Q^{\ast}[V\setminus V_{\text{bad}}]$, with the goal of obtaining a $(1-\eps/2)$-approximate matching, we have,
\begin{align*}
    \Exp[|M|]\geq (1-\frac{\eps}{2})\cdot\Exp[|\MM(Q^\ast[V\setminus V_{\textup{bad}}])|]\geq (1-\frac{\eps}{2})\cdot (1-\eps^7)\cdot \Exp[|\MM(Q^\ast)|]\geq (1-\eps)\cdot \Exp[|\MM(Q^\ast)|].
\end{align*}
By the guarantees on $\Exp[|\MM(Q^{\ast})|]$ (from \Cref{lem:DS25}), the approximation ratio follows. The round complexity of the algorithm is $O(\poly({1}/{\epsilon}))$. This is implied by the fact that we run algorithm $\mathcal{B}$ (from \Cref{lem:localapproxmatching}) on $Q^\ast[V\setminus V_{\text{bad}}]$ with $\delta=\frac{\eps}{2}$ and $Q^\ast[V\setminus V_{\text{bad}}]$ has maximum degree $\frac{1}{\epsilon^{10}}$.

\end{proof}

\begin{proof}[Proof of \Cref{thm:MatchingInPolyRrounds} (Bipartite Graphs)]
Let $\alg$ be the candidate algorithm. If $G$ is bipartite graph, then $\alg$ runs $\mathcal{A}'$ to pick $Q$, which satisfies the premise of \Cref{lem:boundingbadvertices} and has an approximation guarantee of $0.73-\eps$. Thus, we have that 
\begin{align*}
    \Exp[|\MM(Q^\ast[V\setminus V_{\textup{bad}}])|]
    &\geq (1-\eps^{7})\cdot \Exp[|\MM(Q^\ast)|]
\end{align*}
Since we run $\mathcal{B}$ on $Q^{\ast}[V\setminus V_{\text{bad}}]$, with the goal of obtaining a $(1-\eps/2)$-approximate matching, we have,
\begin{align*}
    \Exp[|M|]\geq (1-\frac{\eps}{2})\cdot\Exp[|\MM(Q^\ast[V\setminus V_{\textup{bad}}])|]\geq (1-\frac{\eps}{2})\cdot (1-\eps^7)\cdot \Exp[|\MM(Q^\ast)|]\geq (1-\eps)\cdot \Exp[|\MM(Q^\ast)|].
\end{align*}
By the guarantees on $\Exp[|\MM(Q^{\ast})|]$ (from \Cref{lem:DS25} and using the fact that $G$ is bipartite), the approximation ratio follows. The round complexity of the algorithm is $O(\poly({1}/{\epsilon}))$. This is implied by the fact that we run algorithm $\mathcal{B}$ (from \Cref{lem:localapproxmatching}) on $Q^\ast[V\setminus V_{\text{bad}}]$ with $\delta=\frac{\eps}{2}$ and $Q^\ast[V\setminus V_{\text{bad}}]$ has maximum degree $\frac{1}{\epsilon^{10}}$.
\end{proof}

\section{Logarithmic Approximation for Minimum Dominating Set}
\label{sec:mds}

Let $\Delta$ be the maximum degree of base graph $G$ and let $\maxExpDeg$ be the maximum expected degree of $G^\ast$,
\[
    \maxExpDeg = \max_v \Exp[\deg_{G^\ast}(v)]
    .
\]
Note that $\maxExpDeg$
is {smaller} than the expected maximum degree  $\Exp[\max_v \deg_{G^\ast}(v)]$.

We give an algorithm that finds a dominating set $S$ of $G^\ast$, whose expected size is 
an $O(\log{\maxExpDeg})$-approximation to that of the optimal solution $S^\ast$ for $G^\ast$.

\MDSlogdelta*

At a high level, our aim is to mimic the standard greedy algorithm (e.g.~\cite{Johnson73,Lovasz75}) for computing a logarithmic approximation to the minimum dominating set. 
The preprocessing stage computes a ranking of the vertices, such that the $i$th vertex in the ranking, denoted $v_i$, is the one that \emph{in expectation} covers the largest number of vertices in $G^\ast$ not already covered by the previous $i-1$ vertices, $v_1,\ldots,v_{i-1}$. 
Then, after $G^\ast$ is realized, each vertex chooses its neighbor in $G^\ast$ (or itself) with the highest rank, and the set of all chosen vertices is the output of the algorithm. 

Recall that the textbook analysis of the standard greedy algorithm for a deterministic graph assigns to every vertex covered by the $i$th vertex $u$ in the greedy dominating set a cost of $1/d_i$, where $d_i$ is the total number of vertices that $u$ covers (and were not previously covered).   
Then it is shown that for each vertex $v$ in the optimal dominating set $S^\ast$, the sum of costs of all vertices adjacent to $v$ and including $v$ is at most $\sum_{j=1}^{\deg(v)+1} = \ln\deg(v) + O(1)$.

In our algorithm, vertices are ranked according to the expected number of vertices they cover (rather than the actual number).
This presents two challenges. 
On the one hand, there are vertices for which the actual number of vertices they cover is much smaller than expected (thus they have a higher rank than they should); we call such vertices \emph{bad}.
On the other hand, there are vertices for which the actual number of vertices they cover is much larger than expected (thus they have a lower rank than they should); we call such vertices \emph{costly}.
We deal with bad vertices by assigning a cost of zero to the vertices they cover, and showing that the total number of such bad vertices is small (compared to the optimal dominating set size).
We deal with costly vertices by assigning a cost of 1 to the vertices they cover, and showing that the total number of vertices covered by costly vertices is small.

Below we formally present the algorithm and its analysis.

\paragraph{Distributed Algorithm.}
\begin{enumerate}
    \item (\textbf{Preprocessing stage}) Order the vertices of $V$ inductively, as follows. Let $V_0=\emptyset$ and for every $i=1,\dots,n$ let $V_{i}=\{v_1,\dots,v_{i}\}$. Let $\tilde{G}$ be a random hallucination of $G$. For every $i=1,\dots,n$ and vertex $v$, let $\tilde{W}_{i}(v)$ be a random variable that represents the set of vertices in $\tilde{G}$ that are covered by $v$ which are not already covered by vertices in $V_{i-1}$. 
    The probability of a vertex $u\in N_G(v)$ to be in $\tilde{W}_{i}(v)$ is 
    $p_{uv} \cdot \prod_{w\in N_G(u)\cap V_{i-1}}{(1-p_{uw})}$, 

    since we need its edge to $v$ to get sampled in $\tilde{G}$ and its edges to any other vertex in $V_i$ to not get sampled. Similarly, the probability of $v$ to be in $\tilde{W}_{i}(v)$ is 
    $\prod_{w\in N_G(v)\cap V_{i-1}}(1-p_{vw})$.
    Let $\tilde{w}_{i}(v) = \Exp[|\tilde{W}_{i}(v)|]$.
    Then 
    \[
        \tilde{w}_{i}(v) = \prod_{w\in N_G(v)\cap V_{i-1}}(1-p_{vw})+\sum_{u\in N_G(v)\setminus V_{i-1}} p_{uv}
        \cdot
        \prod_{w\in N_G(u)\cap V_{i-1}}{(1-p_{uw})}.
    \]
    Let $v_i$ be a vertex $v$ that maximizes $\tilde{w}_{i}(v)$.\footnote{Throughout the algorithm and analysis, we break ties in any arbitrary consistent manner, say, by vertex identifiers.} We say that $i$ is the rank of vertex $v_i$. Note that if $i<i'$ then $\tilde{w}_{i}(v_i)\geq \tilde{w}_{i}(v_{i'})\geq \tilde{w}_{i'}(v_{i'})$.
    
    \item (\textbf{Communication round}) Each vertex $v$ selects the vertex $u\in N_{G^\ast}(v)\cup \{v\}$ whose rank is minimal and informs $u$ that it is selected.
    All selected vertices join the dominated set $S$ that the algorithm outputs.
\end{enumerate}

It is immediate that $S$ is indeed a dominating set of $G^\ast$, and that the algorithm has the round and message complexity stated in \cref{thm:MDSlogdelta}. 
Next we bound the approximation ratio.

For every vertex $v\in V$, let $W_i(v)$ be the set of vertices in $G^\ast$ that are covered by $v$ which are not already covered by vertices in 
$V_{i-1} = \{v_1,\ldots,v_{i-1}\}$, and let $w_i(v) = |W_i(v)|$.

In the following we assume without loss of generality that $\bar\Delta$ is larger than a sufficiently large constant (as we can artificially increase $\bar\Delta$ by adding to $G$ a disjoint star graphs with $\bar\Delta$ leaves and edge realization probabilities that are equal to 1).  

\begin{definition}[Bad Vertices and Costly Vertices]
    We say that vertex $v_i$ is \emph{bad} if 
    \[
        w_{i}(v_i) < (\tilde w_{i}(v_{i}) - 8\ln \maxExpDeg)/4.
    \]
    We say that $v_i$ is \emph{costly} if there is $i'< i$ such that 
    \begin{align}
    \label{eq:ugly-condition}
        w_{i'}(v_i) > 6\cdot(\lceil\tilde w_{i'}(v_{i'})\rceil + \log \maxExpDeg ),.
    \end{align}
    and denote $\nu_i$ the smallest such $i'$.

    Let $B$ and $\Ugly$ denote the set of bad vertices and the set of costly vertices, respectively.
\end{definition}

First we  show that the expected number of bad vertices is small.

\begin{lemma}[There are not Many Bad Vertices]
    \label{lem:bad}
    \[
        \Exp[|B|]
        \leq
        n / \maxExpDeg^2
        .
    \]
\end{lemma}
\begin{proof}
    Recall that vertex $v_i$ is bad if 
    $w_{i}(v_i) < (\tilde w_{i}(v_{i}) - 8\ln \maxExpDeg)/4$.
    Suppose first that $\tilde w_{i}(v_{i}) \geq 8\ln \maxExpDeg$.
    By a Chernoff bound
    \begin{align*}
        \Pr(w_{i}(v_i) < (\tilde w_{i}(v_{i}) - 8\ln \maxExpDeg)/4)    
        &\leq
        \Pr(w_{i}(v_i) < \tilde w_{i}(v_{i})/4) 
        \\&
        \leq
        e^{-(3/4)^2\tilde w_{i}(v_{i})/2}
        \\&
        <
        1/\bar\Delta^2
        ,
    \end{align*}
    since $\tilde w_{i}(v_{i}) \geq 8\ln \maxExpDeg$.
    If $\tilde w_{i}(v_{i}) < 8\ln \maxExpDeg$, the above inequality holds trivially.
    Thus, it follows that $\Exp[|B|]\leq n/ \maxExpDeg^2$.
\end{proof}

Next we show that costly vertices do not cover many vertices in expectation.
\begin{lemma}[Costly Vertices do not Cover Many Vertices]
    \label{lem:costly} Let $d_v = \deg_{G^\ast}(v)$ for any $v\in V$. 
    \[
        \Exp\left[\sum_{v \in \Ugly} d_v\right]
        \leq
        n/ \maxExpDeg^3
        .
    \]
\end{lemma}    
\begin{proof}
    We define the set 
    \[
        K = \{\lceil\tilde w_{i}(v_{i})\rceil\colon 1\leq i\leq n\}
        \subseteq\{1,\ldots,\lceil\bar\Delta\rceil\},
    \]
    and for each $k\in K$ we let $j_k$ be the smallest $i$ with $\lceil\tilde w_{i}(v_{i})\rceil = k$. Note that $\lceil\bar\Delta\rceil\in K$ and $j_{\lceil\bar\Delta\rceil} = 1$.
    Observe also that if $v_i\in \Ugly$ then $\nu_i = j_k$ for some $k\in K$.
    Indeed, $\nu_i$ cannot take a value $i'\notin \{j_k\colon k\in K\}$, since otherwise for the $k$ such that $\lceil\tilde w_{i'}(v_{i'})\rceil = k$, we have that $j_k < i'$ and $j_k$ also satisfies \Cref{eq:ugly-condition}, as replacing $i'$ by $j_k$ leaves the right side of inequality \cref{eq:ugly-condition} unchanged and can only increase the left side.

    To simplify notation, we define $\nu_i = 0$ if $v_i \notin \Ugly$, so that now each $v_i\in V$ is associated with a value $\nu_i$.
    We have 
    \[
        \Exp\left[\sum_{v_i \in \Ugly} d_{v_i}\right]
        =
        \sum_{v_i \in V} \Exp\left[d_{v_i}  \cdot \mathds{1}_{v_i \in \Ugly} \right].
    \]
    For each $v_i\in V$,
    \[
        \Exp\left[d_{v_i}\cdot \mathds{1}_{v_i \in \Ugly}\right]
        =
        \sum_{k \in K\colon k\geq \tilde w_{i}(v_i)}
        \Exp[d_{v_i} \cdot \mathds{1}_{\nu_i = j_k}].
    \] 
    For any $k \in K\setminus\{\lceil\bar\Delta\rceil\}$ with $k\geq \tilde w_{i}(v_i)$,
    \begin{align*}
        \Exp[d_{v_i} \cdot \mathds{1}_{\nu_i = j_k}]
        &
        =
        \Exp[d_{v_i} \mid \nu_i = j_k]\cdot\Pr(\nu_i = j_k)
        \\&
        \leq
        6(\lceil\bar\Delta\rceil + \log\bar\Delta)\cdot\Pr(\nu_i = j_k)
        \\&
        \leq
        6(\lceil\bar\Delta\rceil + \log\bar\Delta)\cdot\Pr(w_{j_k}(v_i) 
        > 6(k+\log\bar\Delta)),
    \end{align*}
    where the first inequality holds because 
    $\nu_i = j_k$ and $k\neq\lceil\bar\Delta\rceil$ imply $d_{v_i} = w_1(v_i) \leq 6(\lceil\bar\Delta\rceil + \log\bar\Delta)$, and the second inequality holds because 
    $\nu_i = j_k$ implies $w_{j_k}(v_i) > 6(k+\log\bar\Delta)$.
    Similarly, for $k = \lceil\bar\Delta\rceil$, we have $j_k = 1$ and
    \begin{align*}
        \Exp[d_{v_i} \cdot \mathds{1}_{\nu_i = j_k}]
        &
        =
        \Exp[d_{v_i} \cdot \mathds{1}_{d_{v_i} > 6(\lceil\bar\Delta\rceil + \log\bar\Delta)}]
        \\&
        \leq
        \sum_{j \geq 6(\lceil\bar\Delta\rceil+\log\bar\Delta)}\Pr(d_{v_i} > j)
        \ +\ 6(\lceil\bar\Delta\rceil+\log\bar\Delta)\cdot \Pr(d_{v_i} 
        \geq 6(\lceil\bar\Delta\rceil+\log\bar\Delta))
        ,
    \end{align*}
    where the first equation holds because $\nu_i = j_k = 1$ implies $d_{v_1} = w_1(v_i) > 6(\lceil\bar\Delta\rceil+\log\bar\Delta)$, and the second equation is obtained using \cref{lem:expectation}.

    Next we bound the probabilities quantities in the last two equation above.
    We have $\Exp[w_{j_k}(v_i)] \leq \tilde w_{j_k}(v_{j_k}) \leq k$, thus for any $j\geq 6(k+\log\bar\Delta) > 6\cdot \Exp[w_{j_k}(v_i)]$, a Chernoff bound gives
    \[
        \Pr(w_{j_k}(v_i) \geq j) \leq 2^{-j}
        .
    \]
    In particular, for $k=\lceil\bar\Delta\rceil$ the above yields a bound for $\Pr(d_{v_i} \geq j)$.
    Substituting to the previous two equations (and performing some simple calculations) gives
    \[
        \Exp[d_{v_i} \cdot \mathds{1}_{\nu_i = j_k}]
        < 
        1/\bar\Delta^4,
    \]
    for both cases, for large enough $\bar\Delta$.
    
    Finally, combining the above we obtain
    \begin{align*}
        \Exp\left[\sum_{v_i \in \Ugly} d_{v_i}\right]
        &=
        \sum_{v_i \in V} \Exp\left[d_{v_i}  \cdot \mathds{1}_{v_i \in \Ugly} \right].
        \\&
        =
        \sum_{v_i \in V}         \sum_{k \in K\colon k\geq \tilde w_{i}(v_i)}
        \Exp[d_{v_i} \cdot \mathds{1}_{\nu_i = j_k}]
        \\&
        \leq
        n\cdot \lceil\bar\Delta\rceil \cdot 1/\bar\Delta^4\\
        &<
        n/\bar \Delta^3
        .
    \end{align*}
    This completes the proof of \cref{lem:costly}.
\end{proof}

We can now prove the approximation guarantee stated in  \cref{thm:MDSlogdelta}.
\begin{proof}[Proof of \Cref{thm:MDSlogdelta}]
    As mentioned earlier, $S$ is a dominating set of $G^\ast$ by construction, and the algorithm has the desired round and message complexity.
    It remains to analyze the size of $S$.

    Let $S^\ast$ be an optimal solution for $G^\ast$. 
    We will show that 
    \[
        \Exp[|S|]
        =
        O(\log\maxExpDeg) \cdot \Exp[|S^\ast|]
        .
    \]
    As in the textbook analysis of the greedy algorithm, we assign to each vertex a \emph{cost} for covering it by $S$, but we use a slightly different cost function.
    We denote by $\cost(u)$ the cost for covering a given vertex $u$. 
    Let $v_i$ be the vertex in $S$ that covers $u$, i.e., $v_i$ has the minimum rank $i$ among all vertices in $N_{G^\ast}(u)\cup\{u\}$.
    Then
    \[
        \cost(u)
        =
        \begin{cases}
            \frac{1}{w_i(v_i)}, & \text{ if } v_i\notin B \\
            0, & \text{ if } v_i\in B
            .
        \end{cases}
    \]
    Then
    \begin{equation} \label{eq:bound_on_s}
        |S| 
        \leq |B| + \sum_{v\in V}\cost(v).
    \end{equation}
    Next we lower bound $|S^\ast|$ in terms of the costs, by distinguishing between non-costly and costly vertices. 
    For each vertex $v\in S^\ast \setminus \Ugly$, by similar reasoning as in the standard proof, we get
    \begin{align*}
        \sum_{u\in N_{G^\ast}(v)\cup\{v\}}\cost(u)
        &\leq
        \sum_{1\leq j \leq |N_{G^\ast}(v)| + 1}\frac1{\max\{1,(j - 8\ln\bar\Delta)/2\}}
        \\&
        =
        O(\log(|N_{G^\ast}(v)|+1) + \log\bar\Delta)
        \\&
        =
        O(\log\bar \Delta),
    \end{align*}
    where for the last equation we used that $|N_{G^\ast}(v)| = O(\bar\Delta)$ since $v\notin \Ugly$.
    For any vertex $v \in S^\ast \cap \Ugly$, we use the trivial bound
    \begin{align*}
        \sum_{u\in N_{G^\ast}(v)\cup\{v\}}\cost(u)
        &\leq
        \deg_{G^\ast}(v) + 1
        .
    \end{align*}
    Combining the above two cases yields
    \[
        \sum_{v\in S^\ast}
        \sum_{u\in N_{G^\ast}(v)\cup\{v\}}\cost(u)
        \leq
        \sum_{v \in S^\ast\cap\Ugly}(\deg_{G^\ast}(v)+1) 
        +
        |S^\ast|\cdot O(\log\bar\Delta).
    \]
    From that and the bound on $|S|$ we computed earlier in \cref{eq:bound_on_s}, we have
    \begin{align*}
        |S| &\leq |B| + \sum_{v\in V}\cost(v)
        \\&
        \leq
        |B| + \sum_{v\in S^\ast}
        \sum_{u\in N_{G^\ast}(v)\cup\{v\}}\cost(u)
        \\&
        \leq
        |B| + \sum_{v \in S^\ast\cap\Ugly}(\deg_{G^\ast}(v) + 1)
        +
        |S^\ast|\cdot O(\log\bar\Delta),
    \end{align*}
    where the second inequality holds because $S^\ast$ is a dominating set and thus must cover all vertices.
    Taking the expectation and applying \cref{lem:costly,lem:bad} gives
    \[
        \Exp[|S|]
        \leq
        n/ \maxExpDeg^2
        +
        2n/ \maxExpDeg^3
        +
        \Exp[|S^\ast|]\cdot O(\log\bar\Delta)
        .
    \]
    To complete the proof we observe that 
    \[
        \Exp[|S^\ast|]
        \geq
        \frac{n - \sum_{v \in S^\ast\cap\Ugly}(\deg_{G^\ast}(v)+1)}
        {6(\lceil\bar\Delta\rceil + \log\Delta)}
        =
        \Omega\left(\frac n{\bar\Delta}\right)
        ,
    \]
    where the first inequality holds because all vertices not covered by costly vertices in $S^\ast$ must be covered by vertices of degree at most $6(\lceil\bar\Delta\rceil + \log\Delta)$; and to obtain the second equation we apply again \cref{lem:costly}.
    Combining the last two equations gives 
    $\Exp[|S|]\leq \Exp[|S^\ast|]\cdot O(\log\bar\Delta)$.
\end{proof}

\subsection*{Acknowledgments}

The authors are grateful to Dimitrios Los for several valuable discussions during the course of this work. 
The authors also thank the organizers of \href{https://www.dagstuhl.de/de/seminars/seminar-calendar/seminar-details/24471}{Dagstuhl Seminar 24471 ``Graph Algorithms: Distributed Meets Dynamic"}, where preliminary discussions regarding this work took place. 

Keren Censor-Hillel was supported in part by the Israel Science Foundation, grant No. 529/23. 
Aditi Dudeja was supported by the Austrian Science Fund (FWF): P 32863-N. 
This project has received funding from the European Research Council (ERC) under the European Union's Horizon 2020 research and innovation programme (grant agreement No 947702).
The project has also received funding from the Inria Associate Team DAME.

{\small
\bibliographystyle{alphaurl}
\bibliography{bibliography}
}

\appendix

\section{Appendix}

The following is a simple generalization of a standard formula for the expectation on an integer non-negative random variable.
\begin{lemma}
    \label{lem:expectation}
    If $X$ is a non-negative integer random variable and $\ell$ a non-negative integer then
    \[
        \Exp[X\cdot \mathds{1}_{X\geq \ell}]
        =
        \sum_{\ell'\geq \ell}\Pr(X>\ell')
        \ +\ \ell\cdot\Pr(X\geq \ell) 
        .
    \]
\end{lemma}

\end{document}